\documentclass[journal]{IEEEtran}
\IEEEoverridecommandlockouts                              

\usepackage{graphicx}

\usepackage{mdframed}
\usepackage{fixmath}
\usepackage{enumitem,kantlipsum}
\usepackage{lipsum}
\usepackage{mathtools}
\usepackage{cuted}
\usepackage{romannum}
\usepackage[usenames,dvipsnames]{xcolor}
\usepackage{tabularx,booktabs}
\usepackage{setspace}

\usepackage{setspace}
\usepackage{calrsfs}
\DeclareMathAlphabet{\pazocal}{OMS}{zplm}{m}{n}
\usepackage{graphicx,dblfloatfix}
\usepackage{graphics} 
\usepackage{epsfig} 
\usepackage{mathptmx}
\usepackage{caption}
\usepackage{subcaption}
\usepackage{times} 
\usepackage{tikz}
\usetikzlibrary{positioning}
\usepackage{amsmath, amssymb}

\usepackage{amsthm}
\usepackage{amsfonts} 
\usepackage{setspace}
\usepackage{scalefnt}
\usepackage{multirow}
\usepackage{textcomp}
\usepackage[english]{babel}
\usepackage{float} 
\usepackage{algorithmicx}
\usepackage[section]{algorithm}
\usepackage{algpascal}
\usepackage{algc}
\usepackage{algcompatible}
\usepackage{algpseudocode}
\let\oldReturn\Return
\renewcommand{\Return}{\State\oldReturn}
\usepackage{mathtools}
\usepackage{bm}
\usepackage{mathptmx}
\usepackage{times} 
\usepackage{mathtools, cuted}
\usepackage{textcomp}
\usepackage[english]{babel}
\usepackage{rotating}
\usepackage{pgfplots}
\pgfplotsset{compat=1.5}
\usepackage{pgfplotstable}
\usepackage{filecontents}

\usepackage{setspace}
\usepackage{scalefnt}

\usepackage{url}

\usepackage{diagbox}
\usepackage{makecell}

\newtheorem{theorem}{Theorem}[section]
\newtheorem{lemma}[theorem]{Lemma}

\newtheorem{defn}[theorem]{Definition}
\newtheorem{cor}[theorem]{Corollary}
\newtheorem{rem}[theorem]{Remark}
\newtheorem{prop}[theorem]{Proposition}

\newcommand{\A}{\pazocal{A}}

\renewcommand{\S}{\pazocal{S}}

\newcommand{\overbar}[1]{\mkern 1.5mu\overline{\mkern-1.5mu#1\mkern-1.5mu}\mkern 1.5mu}

\renewcommand{\P}{\pazocal{P}}

\newcommand{\sT}{\scriptscriptstyle{T}}
\newcommand{\sR}{\scriptscriptstyle{R}}
\newcommand{\sFP}{\scriptscriptstyle{FP}}

\newcommand{\C}{\pazocal{C}}
\newcommand{\D}{\pazocal{D}}
\newcommand{\R}{\mathbb{R}}
\newcommand{\Y}{\pazocal{Y}}
\newcommand{\Z}{\pazocal{Z}}

\newcommand{\bS}{{\bf S}}

\newcommand{\bhS}{{\bf S}}

\newcommand{\sA}{{\scriptscriptstyle{A}}}
\newcommand{\sD}{{\scriptscriptstyle{D}}}

\newcommand{\sB}{{\scriptscriptstyle{B}}}
\newcommand{\G}{{\pazocal{G}}}
\newcommand{\F}{\pazocal{F}}

\graphicspath{{figures/}}
\numberwithin{theorem}{section}

\def \cN{{\cal N}}

\def \*{\star}
\def \10n{\!\!\!\!\!\!\!\!\!\!}

\def \sF {{\scriptscriptstyle F}}
\graphicspath{{figures/}}
\usepackage[utf8]{inputenc}
         \DeclareMathAlphabet{\mathscr}{U}{BOONDOX-cal}{m}{n}
         \SetMathAlphabet{\mathscr}{bold}{U}{BOONDOX-cal}{b}{n}
         \DeclareMathAlphabet{\mathbscr} {U}{BOONDOX-cal}{b}{n}

\newcommand{\boundellipse}[3]
{(#1) ellipse (#2 and #3)
}

\begin{document}
\title{\LARGE \bf Dynamic Information Flow Tracking for  Detection of Advanced Persistent Threats: A Stochastic Game Approach}
\author{Shana Moothedath, Dinuka Sahabandu,  Joey Allen, Andrew Clark, \IEEEmembership{~Member, IEEE,}\\ Linda Bushnell,\IEEEmembership{~Fellow, IEEE,} Wenke Lee,\IEEEmembership{~Fellow, IEEE} and Radha Poovendran,\IEEEmembership{~Fellow, IEEE}
    \thanks{S. Moothedath, D. Sahabandu, L. Bushnell, and R. Poovendran are with the Department of Electrical and Computer Engineering, University of Washington, Seattle, WA 98195, USA. \texttt{\{sm15, sdinuka, lb2, rp3\}@uw.edu}.}
    \thanks{J. Allen and W. Lee are with the College of Computing, Georgia Institute of Technology, Atlanta, GA 30332 USA.
    \texttt{jallen309@gatech.edu, wenke@cc.gatech.edu}.}
      	\thanks{A. Clark is with the Department of Electrical and Computer Engineering, Worcester Polytechnic Institute, Worcester, MA 01609 USA. \texttt{aclark@wpi.edu}.}
    }

\maketitle
\thispagestyle{empty}
\pagestyle{empty}

\begin{abstract}
Advanced Persistent Threats (APTs) are  stealthy   attacks by intelligent adversaries. This paper studies the detection of APTs that infiltrate cyber systems and compromise specifically targeted data and/or infrastructures. Dynamic information flow tracking is an information trace-based detection mechanism against APTs that tags suspicious information flows in the system and performs security analysis for  unauthorized use of tagged data. In this paper, we develop an analytical model for resource-efficient  detection of APTs using an information flow tracking game. The game  is a nonzero-sum, turn-based, stochastic game with asymmetric information as the defender cannot distinguish whether an incoming flow is malicious or benign. The payoff functions of the game capture the cost for performing security analysis and the rewards and penalties received by the players.
We analyze equilibrium of the  game and prove that a Nash equilibrium is given by a solution to the  minimum capacity cut set problem on a flow-network derived from the system. The edge capacities of the flow-network are obtained from the cost of performing security analysis. 
Finally, we implement our algorithm on a  real-world dataset for a data exfiltration attack augmented with false-negative and false-positive rates and compute an optimal defender strategy.
\end{abstract}

\IEEEpeerreviewmaketitle
\section{Introduction}\label{sec:intro}
An advanced persistent threat (APT) is a prolonged and targeted cyber attack in which an intruder gains illicit access to a system  and remains undetected for an extended period of time. The intention of an APT  is  to monitor system activity and continuously mine highly sensitive data rather than causing damage to the system or organization. 
APT attacks consist of multiple stages that are  initiated by  an initial compromise and reconnaissance stage to establish a foothold in the system. Attackers then progress through the system, exploring and planning an attack strategy to obtain the desired data. This is followed by  exfiltration of sensitive data, which  is continued over a long period of time. 
Defending against APTs is a challenging task since APTs are specifically designed to evade conventional security mechanisms such as firewalls, anti-virus software and intrusion-detection systems that rely on signatures and can, therefore, guard only against known threats. However, APTs introduce information flows in the form of data-flow commands and control-flow commands while interacting with the system and these are continuously recorded in the  log file of the system.

 After an APT attack, when the consequences of the attack has been identified, a forensic investigation is  conducted using the data from the log files. The purpose of this postmortem investigation is to understand the  incident’s root cause, and construct  appropriate defense strategies against such APTs  and its variants  \cite{JiLeeDowWanFazKimOrsLee-17}. During the forensic analysis, the system log data is  analyzed in an offline setting.

Dynamic Information Flow Tracking (DIFT) \cite{NewSon-05} is a widely used detection mechanism for offline analysis of APT. DIFT  uses the information traces recorded in the  system log for performing the security analysis \cite{JiLeeDowWanFazKimOrsLee-17}. The key idea behind DIFT is that it tags all suspicious input/data channels and tracks the propagation of the  tagged  information flows through the system. DIFT generates security analysis using a pre-specified set of heuristic security rules  whenever it observes an unauthorized use of tagged data. These heuristic rules are typically defined by practicing experts based on the historical attack knowledge
 and  knowledge about  nominal system behavior \cite{JiLeeDowWanFazKimOrsLee-17}.

Although,  security rules incorporated in the DIFT mechanism cover a wide range of attacks, these security rules may not be capable of verifying the authenticity of information flows against all possible attacks,  resulting in the generation of false-negatives and false-positives.  Incorrectly identifying a benign (nominal user) flow  as malicious is a false-positive $(FP)$   and failing to identify a malicious (adversarial) flow is a  false-negative $(FN)$.   For instance, while the security rules for buffer overflow protection \cite{DalKanKoz-08} can be verified accurately, the security rules for web application vulnerabilities \cite{DalKozZel-09} cannot be accurately  verified.  Consequently,  attacks that exploit web application  may lead to incorrect conclusions by DIFT thereby generating $FP$s and $FN$s.

Additionally, limited availability of resources for defense along with the performance and memory overhead imposed by the defense mechanism on the system demands a resource-efficient detection technique. An analytical model of DIFT  that captures the cost for performing security analysis and the  effectiveness of flow-tagging mechanisms will facilitate the trade-off between the effectiveness of defense and the resource efficiency. Further, such a model would enable the design of optimal security strategies while taking into account the generation of $FP$ and $FN$.

In this paper, we provide  an analytical model to enable DIFT to optimally select locations in the system to perform security analysis so as to maximize the probability of detection  while minimizing the cost of detection. 
Our framework is based on the following insights. First, although both APT and DIFT are unaware of the other player's strategy, the effectiveness of the DIFT depends on the APT's strategy and the APT's probability of evading detection depends on the DIFT's strategy. This strategic interaction motivates a game-theoretic approach. Second, the  efficiency of detection also depends on the effectiveness of performing security analysis at different locations in the system, and hence is determined by  rates of $FP$ and $FN$. Third, the game unfolds at multiple states between the entry points  and the exit points of the attack where each state corresponds to the position of the tagged  flows in the system. At each state, the defender decides  whether to analyze a tagged flow and which one of the tagged flows to analyze (to avoid generation of $FP$ and $FN$) while the adversary decides which process to transition to, at the cost of spending the defense resources.  We formulate a stochastic game model that is played on an {\em information flow graph} (IFG). An IFG is a directed graph that expresses the history of a
system's execution in terms of the spatio-temporal relationships between
processes and subjects.
The contributions of this paper are the following.
\begin{enumerate}
\item[$\bullet$] We model the interaction of  APT and  DIFT  as a nonzero-sum, two-player, turn-based stochastic game (${\bf G}$) with finite state and action spaces. In the APT vs. DIFT game, the state of the game is the nodes of the IFG at which the tagged flows arrive. The location of the tagged flows is known to both players (APT and DIFT), and this makes the state of the game observable. However, each player is unaware of the other player’s strategy. Also, before performing security analysis, DIFT cannot distinguish malicious and benign flows. This results in an asymmetric information structure. 
\item[$\bullet$] We analyze Nash equilibrium (NE) strategy  of the  game and show that an NE  can be obtained from a solution to the {\em minimum capacity cut-set} problem on a flow-network constructed from the  IFG of the  system.
\item[$\bullet$] We implement our algorithms  and results on real-world data obtained for a data exfiltration attack using the Refinable Attack INvestigation (RAIN) system \cite{JiLeeDowWanFazKimOrsLee-17}  augmented with $FN$ and $FP$ rates of  the system.
\end{enumerate}

The rest of the paper is organized as follows. Section~\ref{sec:rel} presents  related work. 
Section~\ref{sec:prelim} introduces  preliminary concepts on information flow graph,  APT attack, and a description of DIFT-based defense mechanism. Section~\ref{sec:prob} describes the formulation of the turn-based stochastic game model. 
Section~\ref{sec:comp} gives the solution concept of the game and the equilibrium analysis results for computing optimal strategies of the players. Section~\ref{sec:sim} illustrates the numerical results using data exfiltration attack data set collected using RAIN and augmented with $FN$ and $FP$ rates. Section~\ref{sec:conclu} concludes the paper.
\section{Related Work}\label{sec:rel}
Stochastic games model  interactions of multiple agents that jointly control the evolution of states of a stochastic dynamical system \cite{Sha-53}. 
Stochastic games have been widely used  to model security games  \cite{LyeWin-05},  economic games \cite{Ami-03}, and resilience of cyber-physical systems \cite{ZhuBas-11}, where  each player tries to maximize its individual payoff. A brief overview of the existing results in stochastic games is given in \cite{JasNow-16}.

Stochastic games have been used  to address system security related problems. The interaction between malicious
attackers  and the intrusion detection system (IDS) is modeled using a stochastic game in  \cite{AlpBas-06}.  A  nonzero-sum stochastic game model is given  in \cite{ZhuTemBas-10} to model network security configuration problem in distributed IDS. 
Then a value iteration based algorithm is proposed in \cite{ZhuTemBas-10} to find an $\epsilon$-NE  for an attacker model where multiple adversaries simultaneously attack a network. 
A zero-sum stochastic game is formulated in \cite{NguAlpBas-09} for IDS in a communication or computer network with interdependent nodes and correlated security assets and vulnerabilities. 
Note that,  \cite{HesPra-01} considered a finite-horizon game and \cite{AlpBas-06}-\cite{NguAlpBas-09} dealt with zero-sum stochastic games for IDS.

Game-theoretic models for resource-efficient detection of APTs are given in \cite{SahXiaClaLeePoo-18}, \cite{MooSahClaLeePoo-18}, \cite{MooShaAllVClaBushWenPoo-18_arx}, and \cite{dinukaCDC}. DIFT models with fixed trapping nodes  are introduced and analyzed in \cite{SahXiaClaLeePoo-18},  \cite{MooSahClaLeePoo-18}. Paper  \cite{MooShaAllVClaBushWenPoo-18_arx} extended the models in \cite{SahXiaClaLeePoo-18}, \cite{MooSahClaLeePoo-18} by considering a DIFT model which selects the trapping nodes in a dynamic manner rather than being fixed and proposed a min-cut based solution approach. Reference \cite{dinukaCDC} considered the detection of APTs when the attack consists of multiple attackers  possibly with different capabilities and analyzed the best responses of the players.   Note that, the focus in \cite{SahXiaClaLeePoo-18}, \cite{MooSahClaLeePoo-18}, and \cite{MooShaAllVClaBushWenPoo-18_arx} is resource-efficient detection of APTs and they do not consider the false-positives and false-negatives generated in the system. In other words, the game models in \cite{SahXiaClaLeePoo-18}-\cite{dinukaCDC}  assume that security analysis performed by DIFT can verify the authenticity of  tagged information flows accurately. 

A stochastic game model for detecting APTs was recently introduced in the conference version of this work \cite{Dinuka_ACC-19}.  The proposed method in \cite{Dinuka_ACC-19} analyzed a discounted stochastic game and presented a value iteration based algorithm to obtain an $\epsilon$-NE of the discounted game. Due to the discounted nature of the payoff, the Nash equilibrium analysis of the game in \cite{Dinuka_ACC-19} reduces to a nonlinear program (NLP). Solving an NLP is computationally challenging and the convergence of the algorithm in \cite{Dinuka_ACC-19} is not guaranteed. This necessitates an alternate solution approach for solving the DIFT vs. APT game. In this paper, we  solve an average reward (undiscounted) stochastic game and present a min-cut based approach with guaranteed convergence to compute optimal defense strategies.

In this paper we consider that the trapping nodes are generated dynamically and we extend our previous results on min-cut analysis in \cite{MooShaAllVClaBushWenPoo-18_arx} in the following aspects. (i)~Since there can be a wide range of possible APT attacks, we recognize that the security rules of  DIFT may not be  capable of verifying the authenticity of the information flows accurately thereby resulting in the generation of $FP$ and $FN$. Consequently,  the game model in this paper is stochastic unlike  in \cite{MooShaAllVClaBushWenPoo-18_arx}. (ii)~We introduce a model with multiple information flows out of which one is malicious and the remaining are benign unlike in  \cite{MooShaAllVClaBushWenPoo-18_arx}  which dealt with one information flow. Note that, while DIFT knows there is exactly one malicious flow in the system, it is unaware which flow is the malicious flow before performing security analysis. This results in an asymmetric information between players and hence the game is imperfect information game.  (iii)~In the current game model, although an information flow is concluded as benign at an initial inspection, it can be inspected again later during its propagation.  The model in \cite{MooShaAllVClaBushWenPoo-18_arx} inspects the flow exactly once as $FN$ and $FP$ rates are assumed to be zero. However, DIFT may analyze a flow  multiple times and this will result in added resource cost. The game model considered in this paper captures the added resource cost incurred from reevaluation of the flows.

\section{Preliminaries}\label{sec:prelim}
In this section, we first describe the graphical representation of  the  system, denoted as the information flow graph, and then present the  models of the attacker and the flow tracking-based defense.
\subsection{Information Flow Graph (IFG)}\label{subsec:provenance}
Let $\G = (V_{\G}, E_{\G})$ represent the IFG of the system.  $V_{\G} =\{v_1, \ldots, v_N\}$ consists of the processes (e.g., an instance of a computer program), files,  and objects in the system and $E_{\G} \subseteq V_{\G} \times V_{\G}$ represents the information flows (directed) in the system from one node to the other.   
 IFG-based auditing is heavily desired by large enterprises and government agencies to detect APTs. We perform our game-theoretic analysis  on the IFG of the system and use DIFT as the defense mechanism to detect APTs.
 
 \subsection{Attack Model: Advanced Persistent Threats (APTs)} 
APTs are intelligent attacks with specific targets.  APTs  are stealthy attacks that perform data exfiltration at an ultra-low-rate to avoid detection.  Unlike classical malware, APT campaigns tend to involve multiple hosts, multiple systems, and extend over a long period of time, up to several months~\cite{stuxnet}. APTs are characterized by their abilities to
render existing security mechanisms ineffective.  APTs can evade security protection because existing mechanisms lack  sufficient visibility into user, program, and operating system activities to ascertain  the authenticity of an activity and the provenance of its data. 
%
The timeline and key stages of the APT lifecycle is described below and  presented in Figure~\ref{fig:timeline}.

\begin{enumerate}

\item[1.]  Reconnaissance:  During the reconnaissance phase, the attacker 
   will try to gain information about the system, such as what nodes are
   accessible on the system and the security defenses that
    are being used. 

\item[2.] Initial Compromise: During the initial compromise stage, the attacker's
   goal is to gain access to an enterprise's network. In most cases, the
   attacker achieves this by exploiting a vulnerability or a social engineering 
   trick, such as a phishing email.

\item[3.] Foothold Establishment: Once the attacker has completed the initial
   compromise, it will establish a persistent presence by opening up a
   communication channel with their Command \& Control (C\&C) server. 

\item[4.] Lateral Movement:  The attacker will increase its control of the system
   by moving laterally to new nodes in the system. 

\item[5.] Target Attainment:  Next, the attacker will try and escalate its
   privileges which may be necessary in order to access sensitive
   information, such as a proprietary source code or customer information.

\item[6.]  Attack Completion: The final goal of the attacker is to deconstruct the attack, hopefully in a way to minimize its footprint in order to evade detection. For example, attackers may rely on deleting the system's log.

\end{enumerate}

\begin{figure}[h]
	\centering
	{\includegraphics[scale=0.7, width=0.5\textwidth]{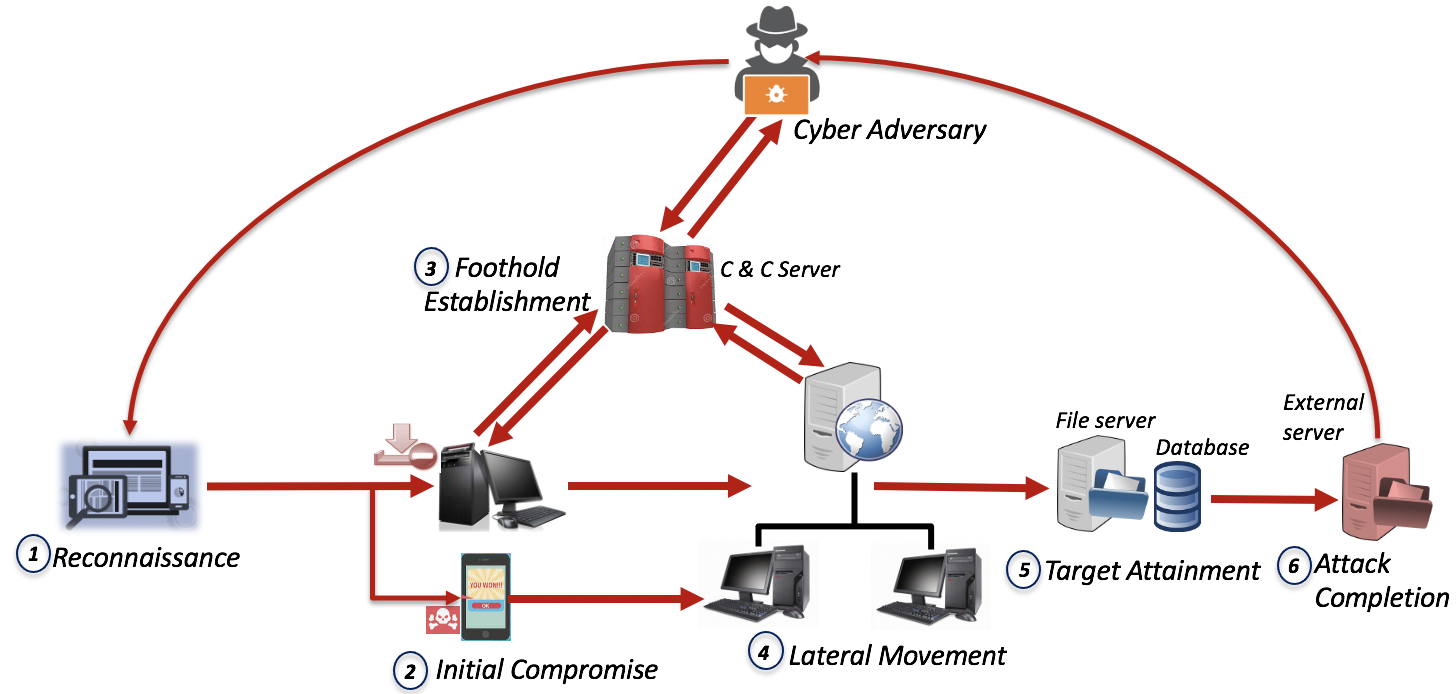}}
	\hspace{\parindent}
	\vspace*{-3 mm}
	\caption{\small Schematic diagram of an APT attack timeline.}\label{fig:timeline}
\end{figure} 

                      
Let the set of possible entry points of the APT be denoted as $\lambda \subset V_{\G}$.  During an attack, the goal of the APT is to capture a subset of nodes of the IFG  referred to as {\em destinations} and denoted as $\D \subset V_{\G}$. Here, $\lambda \cap \D = \emptyset$. Once a foothold is established in the system, APT tries to elevate the privileges and proceeds to the destinations through more internal compromises and   performs data exfiltration at an ultra-low-rate.  In order to achieve this the APT performs operations in the system to transition through the nodes in $V_{\G}$ and arrive at some node in $\D$. In other words, the adversary (APT) selects paths in the IFG and performs transitions along the paths to reach the set $\D$ from the set $\lambda$ .

The  adversary   in the DIFT vs. APT game has the following properties.  We consider an APT attack  that generates a single malicious information flow.  
 The malicious flow  originates at an entry point of the attack and  the objective of the flow is to reach a destination. The state of the game is the set of nodes of the information flow graph at which the tagged flows arrive. APT observes the tagged flows  and hence the state of the game is observable to APT. However, APT is unaware of the actions and the strategy of DIFT. Thus the information structure of APT is asymmetric with respect to that of the defender.

\subsection{Defender Model: Dynamic Information Flow Tracking (DIFT)} 
The objective of the defender is to prevent  any adversarial information flow traversing to the destination nodes. 
DIFT is a dynamic taint analysis based detection mechanism that consists of three main components: (i)~tag sources, (ii)~tag propagation rules, and (iii)~trapping nodes. Trapping nodes are processes and objects (files and network endpoints) in the system that are considered as  suspicious  sources of information by the system. All data originating from a tag source is labeled or tagged and DIFT tracks the propagation of a tagged data through the system.   Propagation of a tagged information flow through the system results in tagging of more information flows based on the  tag propagation rules specified by the DIFT \cite{SuhLeeZhaDev-04}.  
Tag propagation rules are defined  based on two kinds of information flows in the system: explicit information flows and implicit information flows  \cite{ClaLiOrs:07}. 

During the execution of a program in the system, DIFT keeps track of the tagged information flows and generates  trapping nodes for any  unauthorized use of tagged data that indicate a possible attack.  DIFT  invokes security analysis at the trapping nodes using the specified security rules to verify the authenticity of the tagged flow thereby concluding whether there is an attack or not. These security rules  are pre-specified depending on the application running on the system \cite{DalKanKoz-08}. Note that tagged (suspicious) flows  consist of both benign and malicious flows. As the specified security rules may  not necessarily cover all possible attacks by APTs, DIFT generates false-positives and false-negatives.  An optimal selection of trapping nodes in the IFG is hence critical to detect  a wide range of attacks in the system with minimum false-positives and false-negatives. Note that while tag sources and tag propagation rules are known to the attacker, the trapping nodes are dynamically generated during the operation of the system and hence is unknown to the attacker. Figure~\ref{fig:schematic} shows a schematic diagram of the DIFT-based detection framework considered in this paper.

\begin{figure}[h]
	\centering
	{\includegraphics[width=0.5\textwidth]{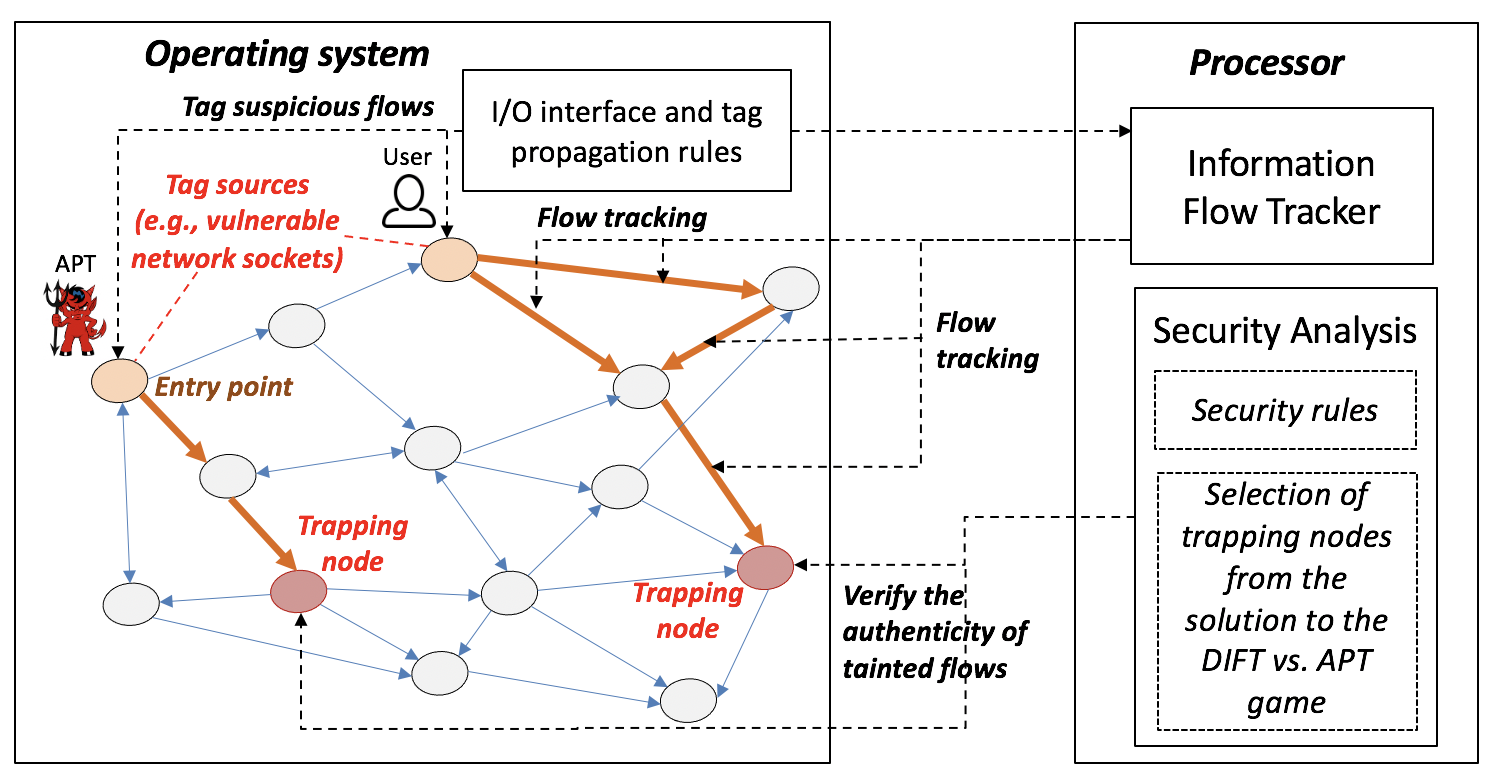}}
	\hspace{\parindent}
	\vspace*{-3 mm}
	\caption{\small Schematic diagram of the DIFT detection framework. The nodes in the figure denote system components such as processes, files, and subjects.}\label{fig:schematic}
\end{figure} 

The  defender   in the DIFT vs. APT game has the following properties.  The state of the game, which is the nodes of the information flow graph at which the tagged flows arrive, is known to DIFT. Thus DIFT observes the state of the game. However, DIFT cannot distinguish a malicious and a benign flow before performing security analysis. Also, while DIFT knows that there is an APT in the system, DIFT is unaware of the actions and the strategy of APT.  

\section{Turn-Based APT vs. DIFT Game  ${\bf G}$}\label{sec:prob}
In this section, we model the interaction of the APT with the system during the different stages of the attack as  a  dynamic stochastic game between the defender (DIFT) and the adversary (APT). Let $\P_{\sD}$ be the defender player and  $\P_{\sA}$ be the adversarial player. We consider a nonzero-sum turn-based game. The information structure of the players is asymmetric as the defender player does not know if a tagged (suspicious) flow is malicious or benign while the adversarial player knows this information. The game ${{\bf G}} = \{\bhS, s_0, {\Sigma},{ \Delta}, {P}\}$ unfolds on a finite state space, $\bhS$, with  initial state, $s_0$, finite  action space of players, ${\Sigma}$,   labeled transitions ${\Delta} \subseteq \bhS \times {\Sigma} \times \bhS$, and transition probability matrix $P$. 

\subsection{State Space}\label{subsec:statespace}

We consider a turn-based game with  the {\em state} of the game at  time $t  \in \{0, 1, 2,\ldots \}$ denoted by the random variable ${s}_t$.  Let $T$ denote the time horizon of the game, i.e., the game ends at time $T$ and $s_{t'} = s_T$, for all $t'\geqslant T$.   The state space ${\bf S}$ is partitioned into two sets, ${\bf S}_{\sA}$ and ${\bf S}_{\sD}$, such that ${\bf S}_{\sA}\cap {\bf S}_{\sD} = \emptyset$ and ${\bf S}_{\sA} \cup {\bf S}_{\sD} = {\bf S}$. The states that belong to the set ${\bf S}_{\sA}$ are referred to as the {\em adversary-controlled} states and the states in ${\bf S}_{\sD}$ are referred to as the {\em defender-controlled} states. Specifically, ${\bf S}_{x}$ is the subset of states at which player $\P_x$, where $x \in \{A, D\}$, controls the transitions.  The state of the game at time $t$  is the state of the system which corresponds to the position of the tagged (suspicious) flows at time $t$.  
 Let $W$ be the number of tagged flows that arrive into the system at time $t=1$. In practice (from the analysis of log data), only a small fraction of nodes in the system receive a system call at the same time. Hence in our model we assume that $W<<N$.  We denote the state of the game at time $t$ using the location of the $W$ flows in the system and  one additional bit which represents whether that state belong to ${\bf S}_{\sA}$ or ${\bf S}_{\sD}$. 

Specifically, at time $t$, $s_t = (x, v_{i_1},  \ldots, v_{i_W})$, where $\{v_{i_1}, \ldots, v_{i_W}\} \in V_{\G} \cup \{ \phi, \tau\}$ and $x \in \{A, D\}$.  Here, $x=A$ implies that $s_t \in {\bf S}_{\sA}$ and $x=D$ implies that $s_t \in {\bf S}_{\sD}$. Moreover, $v_{i_k}=\phi$ if the $k^{\rm th}$ information flow  {\em drops} out, and $v_{i_k}=\tau$ if the $k^{\rm th}$ information flow is  {\em trapped} by DIFT.
  At time $t \geqslant 1$, the state of the game is $s_t = (x, v_{i_1},  \ldots, v_{i_W})$, where $x \in \{A,D\}$ and $\{v_{i_1}, \ldots, v_{i_W}\} \in V_{\G} \cup \{ \phi, \tau\}$. Note that,  out of the $W$ flows one is a malicious flow and the remaining $(W-1)$ are benign flows. For notational convenience we consider the first flow is malicious (note that, defender does not have this information).  While DIFT observes all the $W$ flows, it cannot distinguish  malicious and benign flows. APT, on the other hand, knows which flow is the malicious flow. Thus in the DIFT vs. APT game, APT and DIFT   have asymmetric information.

  Also, corresponding to a state $s_t$ if  $v_{i_k} =\phi$, for some $k \in \{1, \ldots, W\}$, then for all time $t' > t$ the state corresponding to the $k^{\rm th}$ flow remains $\phi$, i.e., if a flow is dropped at time $t$ it remains dropped through out the rest of the game.   

\subsection{Action Spaces}\label{subsec:actionspace}
 We first construct a directed flow-network $\F=(V_{\F}, E_{\F})$ from the IFG $\G$ by introducing a source node $s_{\sF}$ with an outgoing edge to all the entry points and a sink node $t_{\sF}$ with an incoming edge from all the destination nodes.  Here $V_{\F} = V_{\G} \cup \{s_{\sF}, t_{\sF}\}$ and $E_{\F} = E_{\G} \cup \{s_{\sF}\times \lambda\} \cup \{\D \times t_{\sF}\}$. 
\begin{defn}\label{def:attack_path}
 An attack path in the flow-network $\F$ is a simple directed path\footnote{A directed path is said to be a simple directed path if there are no cycles or loops in the path.}  from $s_{\sF}$ to $t_{\sF}$. The set of attack paths in  $\F$ is denoted as $\Omega_{\D}$.
\end{defn} 
The objective of the APT attack  is to capture a destination node in set $\D$. To achieve the same, APT chooses transitions along a path from the set $\lambda$ to the set $\D$ which translates into an attack path in the flow-network $\F$.  Since APT is a stealthy attack, APT performs minimum amount of activities in the system in order to avoid detection. Thus, there  are no cycles  in the transition path of the APT through the IFG. We also note that, any IFG with set of cycles can be converted into an acyclic IFG without losing any causal relationships between the components given in the original IFG. One such dependency preserving conversion is node versioning given in \cite{hossain2018dependence}. Hence  all $\omega \in \Omega_{\D}$ are simple directed paths.

Using the construction of the flow-network, we denote the initial state of the game at $t=0$ is  denoted as $s_0= (A, s_{\sF}, \ldots, s_{\sF})$, since at $t=0$ all flows originate at the source node $s_{\sF}$. At a state $s_t$, either  $\P_{\sD}$ or $\P_{\sA}$ chooses an action from their respective action sets denoted by $\A_{\sD}$ and $\A_{\sA}$, depending on whether $s_t \in {\bf S}_{\sD}$ or $s_t \in {\bf S}_{\sA}$.  If $s_t \in {\bf S}_{\sA}$, the adversary decides whether to quit the attack by dropping the information flow or to continue the attack. If the adversary decides not to quit the attack, then it selects which neighboring node of the IFG to transition from the current node so as to reach set $\D$. In other words, the adversary either drops the malicious flow or performs a transition along a path in $\Omega_{\D}$ and hence $\A_{\sA}:= \{ \phi \}  \cup V_{\G} $.  Specifically, at state $s_t =  (A, v_{i_1}, \ldots, v_{i_W})$, where $v_{i_1} \in V_{\G}$ and $v_{i_1} \notin \D$, $\A_{\sA}(s_t) \in \{ \phi\} \cup \cN(v_{i_1}) $, where $\cN(v_{i_1}) =\{v_j\in V_{\G}: (v_{i_1}, v_j)\in E_{\G}\}$. 
On the other hand,  if $s_t \in {\bf S}_{\sD}$, the action of the defender is to decide whether to perform security analysis on an  information flow or not. If the defender
chooses to perform security analysis, it also selects a flow to analyze.  A node at which security analysis is performed is referred to as a {\em trapping node}. Consider a state $s_t =  (D, v_{i_1}, \ldots, v_{i_W})$ and let defender decides to perform security analysis on the flow at $v_{i_1}$. Then $v_{i_1}$ is chosen as a trapping node.  Thus,  $\A_{\sD}:= \{0\} \cup V_{\G}$, where $0$ represents not analyzing any flow. The defender performs security analysis on at most one flow in a state as there is only malicious flow and also it is not possible to perform analysis on the information flows at the entry points of the attack, $\lambda$. Thus at state $s_t =  (D, v_{i_1}, \ldots, v_{i_W})$,    $\A_{\sD}(s_t) \in \{ 0\} \cup \{v_{i_k} :  v_{i_k} \notin \lambda, k \in \{1, \ldots, W\}\}$.   We also note that the information flow chosen by the DIFT to perform security analysis can be either the malicious flow or a benign flow.

 A state $s_t $ is said to be an {\em absorbing state} if the game terminates at $s_t$ and the action sets of the players are empty, i.e., $\A_{\sA}(s_t) = \A_{\sD}(s_t) = \emptyset$. In ${\bf G}$ a state $s_t = (x, v_{i_1}, \ldots, v_{i_W})$, where $x \in \{A, D \}$,  is absorbing if  one of the  three cases hold: 
\begin{itemize}
\item[(i)]~$v_{i_1} \in  \D \subset V_{\G}$,
\item[(ii)]~$v_{i_k} = \tau$ for some $k \in \{1, \ldots, W\}$, 
\item[(iii]~$v_{i_k} =\phi$ for all $k \in \{1, \ldots, W\}$. 
\end{itemize}
 
Case~(i) corresponds to adversary reaching a destination, case~(ii) corresponds to the case where DIFT performed security analysis on an information flow and also concluded it as malicious, i.e., a flow is trapped, and case~(iii) corresponds to all $W$ flows dropping out. 

\subsection{State Transitions}\label{subsec:tran}
\begin{defn}\label{def:terminate}
Time $t$ is said to be the termination time of the game ${\bf G}$ if it satisfies one of the following condition: (a)~$t=T$ and  (b)~the state of the game $s_t = (x, v_{i_1}, \ldots, v_{i_W})$, where $x \in \{A, D \}$, is an absorbing state.
\end{defn}

We use the notation $W$  to denote the number of tagged information flows that are recorded at the first instant of time, i.e., $t=1$. The log recording system (RAIN) used in our work uses the system call in \cite{log} to create a timestamp for each audit log.  At every instant of the log recording, RAIN records at most one system call for each system component (processes, files). As a result,  the $W$ information flows arrive at distinct nodes of the information flow graph as it is not possible for multiple flows to arrive at one node at the same instant of time.  Let the state of the game at $t=1$ be  $s_t = (D, v_{i_1},  \ldots, v_{i_W})$. Here $\{v_{i_1}, \ldots, v_{i_W} \}$ are the nodes of the IFG at which the flows arrive. The DIFT-based defense mechanism observes the flows. The player $\P_{\sD}$ now chooses an action from the set $\A_{\sD} (s_t) \in \{0, v_{i_1}, \ldots, v_{i_W} \}$. At state $s_t$,  there are three possibilities: (1)~$\P_{\sD}$ does not analyze any flow, (2)~$\P_{\sD}$ analyze the malicious flow, and (3)~$\P_{\sD}$ analyze a benign flow.  Based on the action chosen by $\P_{\sD}$, i.e., $\A_{\sD}(s_t)$, the next possible state of the game $s_{t+1}$ is defined.  At $s_{t+1}$ $\P_{\sA}$ has two possibilities to choose from: (a)~to quit the flow, $\phi$ and (b)~to transition to an out-neighbor of node $v_{i_1}$, $\cN( v_{i_1})$,  i.e., $\A_{\sA} (s_{t+1}) \in \{\phi\}\cup \cN( v_{i_1})$. Based on the action of $\P_{\sA}$ and the distribution of the benign flows in the system,  the next state $s_{t+2}$ is arrived. At $t+2$, $\P_{\sD}$ again chooses its action and the game continues in a turn-based fashion. Note that the DIFT will analyze only one information flow at a time $t$ as we consider a single malicious flow in the system.

For case~(1) the action of $\P_{\sD}$ is $\A_{\sD} (s_t) = 0$, i.e.,  DIFT does not perform security analysis on any of the $W$ information flows. Then the next state of the game $s_{t+1} = (A, v_{i_1}, \ldots, v_{i_W})$. 
For case~(2),  $\P_{\sD}$  chooses action $\A_{\sD} (s_t) = v_{i_1}$ and the next state of the game is $s_{t+1} = (A, v_{i_1},  \ldots, v_{i_W})$ with probability $FN$ and $s_{t+1} = (A, \tau, v_{i_2},  \ldots, v_{i_W})$ with probability $1-FN$.  For case~(3), the player $\P_{\sD}$  chooses action $\A_{\sD} (s_t) = v_{i_k}$,  $k \neq 1$, i.e., $k \in \{2, \ldots, W\}$. Then the 
next state of the game is $s_{t+1} = (A, v_{i_1}, \ldots, v_{i_{k-1}}, \tau,  v_{i_{k+1}}, \ldots, v_{i_W})$ with probability $FP$ and $s_{t+1} = (A, v_{i_1}, \ldots, v_{i_W})$ with probability $1-FP$. Here, $FN$ and $FP$ are the rate of false-negatives and false-positives in the DIFT-architecture which are empirically computed. The values of $FP$ and $FN$ depend on the security rules and  vary across the different DIFT-architectures and they determine the transition probabilities ${P}$ of the stochastic game. In order to reduce false-positives and false-negatives in the system, the security rules  must be capable of identifying the behavior and predicting the intend  of information flows, i.e., determine whether an unknown flow is indeed malicious, or whether it is benign flow that is exhibiting malware-like behavior  \cite{YinSonEgeChrEng-07}.

If $t+1$ is not a terminating time instant, then the adversarial player $\P_{\sA}$ chooses action $\A_{\sA} (s_{t+1}) \in \{\phi\} \cup \cN(v_{i_1}) $. The remaining $W-1$ flows follow the benign flow distribution in the system, which is computed empirically from the nominal system operation and known. Let $\pi_{\sB} : v_i\in V_{\G} \rightarrow [0, 1]^{|\cN(v_{i})|+1}$ denote the benign flow distribution.  If the action of $\P_{\sA}$ is $\phi$, then $s_{t+2} = (D, \phi, v_{j_2}, \ldots, v_{j_W})$, where $v_{j_2}, \ldots, v_{j_W}$ depend on the distribution $\pi_{\sB}$. Here $\{v_{j_2}, \ldots, v_{j_W}\} \in \{\phi\} \cup V_{\G}$.  Note that a benign flow can also drop out. Now $\P_{\sD}$ chooses its action and the game continues.  To summarize, given $s_t = (D, v_{i_1}, \ldots, v_{i_W}) $ (w.p stands for with probability), state transitions and the corresponding transition probabilities, given by ${P}$, are 
{\scalefont{0.95}{
 \begin{equation}\label{eq:transition_D}
 s_{t+1} =
 \begin{cases}
 \begin{array}{lll}
\hspace*{-2.5 mm} {\scalefont{0.2} (A, v_{i_1}, \ldots, v_{i_W}) }, &\hspace*{-5 mm}  \mbox{~w.p~} 1,  &\hspace*{-3 mm}\mbox{if~} {\scalefont{0.8}\A_{\sD}(s_t)=0},\\
\hspace*{-2.5 mm}    (A, v_{i_1}, \ldots, v_{i_W}), &\hspace*{-5 mm} \mbox{~w.p~} FN, &\hspace*{-3 mm}  \mbox{if~}  \A_{\sD}(s_t)=v_{i_1},\\
\hspace*{-2.5 mm}     (A, \tau, \ldots, v_{i_W}),  &\hspace*{-5 mm} \mbox{~w.p~} 1-FN, &\hspace*{-3 mm}  \mbox{if~}  \A_{\sD}(s_t)=v_{i_1},\\
\hspace*{-2.5 mm}  {\scalefont{0.8}(A, v_{i_1}, \ldots, v_{i_W})}, &\hspace*{-5 mm} \mbox{~w.p~} {\scalefont{0.85}1-FP}, &\hspace*{-3 mm} \mbox{if~}  {\scalefont{0.8}\A_{\sD}(s_t)=v_{i_k}, k\neq 1}\\
\hspace*{-2.5 mm} (A, v_{i_1}, \ldots, \tau,  \ldots, v_{i_W}), &~\hspace*{-5 mm} \mbox{~w.p~} FP, & \hspace*{-3 mm} \mbox{~if~}  \A_{\sD}(s_t)=v_{i_k}, k\neq 1.
 \end{array}
 \end{cases}
 \end{equation}
 }}
and
 {\scalefont{0.95}{
  \begin{equation}\label{eq:transition_A}
 s_{t+2} =
 \begin{cases}
 \begin{array}{ll}
{\scalefont{0.9} (D, v_{j_1}, \ldots, v_{j_W})},  &\mbox{~if~}{\scalefont{0.9} \A_{\sA}(s_{t+1})\in \cN(v_{i_1})},\\
{\scalefont{0.9}(D, \phi, v_{j_2}, \ldots, v_{j_W})},& \mbox{~if~}  {\scalefont{0.9}\A_{\sA}(s_{t+1})=\phi}.
 \end{array}
 \end{cases}
 \end{equation}
 }}

\subsection{Strategies of the Players}\label{subsec:stochastic_policy}
A strategy  is a rule that each player uses to select actions at every step of the game. We consider {\em mixed} (stochastic) and {\em behavioral}  player strategies.  
Since the  strategy is mixed, at a  state, $\P_{\sD}$ and $\P_{\sA}$ select an action from the action set $\A_{\sD}$ and $\A_{\sA}$, respectively, based on some probability distribution.
Further, since the strategy is behavioral, the probability distribution on the action set at a time instant $t$ depends  on all the states traversed  by the game and the actions taken by the players until time $t$.   

Let $d_t, a_t$ be the action of the defender and the attacker, respectively, at time $t$. We denote the information available to the  players  $\P_{\sD}$ and $\P_{\sA}$ at time $t \leqslant T$ by ${\bf {Y}}_t$ and ${\bf {Z}}_t$, respectively. Then
\begin{eqnarray}
{\bf {Y}}_t &:=& \{ {s}_0, d_0, {s}_1, d_1, \ldots, {s}_{t-1}, d_{t-1}, {s}_t\}, \nonumber \\
{\bf {Z}}_t &:=& \{s_0, a_0, {s}_1, a_1, \ldots, {s}_{t-1}, a_{t-1}, {s}_t\}.
\end{eqnarray}
We denote the set of all possible outcomes for ${\bf {Y}}_t$ and ${\bf {Z}}_t$ at time $t$ using ${\Y}^{\*}$ and ${\Z}^{\*}$, respectively. 
Let the set of all  pure  strategies of $\P_{\sD}$ and  $\P_{\sA}$ for game ${\bf G}$ be ${\bar{{\bf P}}}_{\sD}$ and ${\bar{{\bf P}}}_{\sA}$, respectively. Define $\Delta{\bar{{\bf P}}}_{\sD}$ ($\Delta{\bar{{\bf P}}}_{\sA}$) as the simplex of ${\bar{{\bf P}}}_{\sD}$  (${\bar{{\bf P}}}_{\sA}$) or, the set of probability distributions over ${\bar{{\bf P}}}_{\sD}$ (${\bar{{\bf P}}}_{\sA}$). A mixed strategy for $\P_{\sD}$ is an element $p_{\sD} \in \Delta{{\bar{\bf P}}}_{\sD}$, so that $p_{\sD}$ is a probability distribution over ${\bar{{\bf P}}}_{\sD}$. Similarly, a mixed strategy for $\P_{\sA}$ is an element $p_{\sA} \in \Delta{\bar{{\bf P}}}_{\sA}$, so that $p_{\sA}$ is a probability distribution over ${\bar{{\bf P}}}_{\sA}$. A behavioral mixed strategy of player $\P_{\sD}$ is given by ${p}_{\sD}: {\Y}^{\*} \rightarrow \Delta{\bar{{\bf P}}}_{\sD}$ and of player $\P_{\sA}$ is given by ${p}_{\sA}: {\Z}^{\*} \rightarrow \Delta{{\bar{\bf P}}}_{\sA}$.
\subsection{Payoffs to the Players}\label{subsec:payoff}
The payoff functions of $\P_{\sD}$ and $\P_{\sA}$ are denoted by $U_{\sD}$ and $U_{\sA}$, respectively. $U_{\sD}$ consists of  three components: (i)~resource cost $\C_{\sD}(v_i)<0$ for performing security analysis of a flow at a node $v_i \in V_{\G}$, (ii)~penalty $\beta_{\sD} < 0$ for adversary reaching a destination node in  $\D$, and (iii)~reward $\alpha_{\sD} > 0$ for detecting  the adversary. Similarly, $U_{\sA}$ consists of  two components: (i)~a reward  $\beta_{\sA}>0$ for reaching a destination node in the set $\D$, and (ii)~penalty $\alpha_{\sA} < 0$ for getting detected by the defender. The reward and penalty parameters and the resource cost of DIFT and APT do not necessarily induce a zero-sum scenario where the payoff that is gained by one player is lost by the other. Therefore, we considered a non-zero sum payoff structure for the DIFT vs. APT game.

Let the payoff of  player $\P_x$ at an absorbing  state ${s}_t$ be denoted as ${c}^x({s}_t)$, where $x \in \{A, D\}$. Also, let the payoff of $\P_{\sD}$ at a non-absorbing defense-controlled state ${s}_t$ with action $d_t$ be ${r}^{\sD}(s_t,  d_t)$, and the payoff of $\P_{\sA}$ at a non-absorbing adversary-controlled state ${s}_t$ with action $a_t$ be ${r}^{\sA}(s_t,  a_t)$.  At each state in the game, $s_t$ at time $t$, where $t<T$ and $s_t$ is a non-absorbing state,  player chooses its action  ($d_t$ for $s_t \in {\bf S}_{\sD}$ and $a_t$ for $s_t \in {\bf S}_{\sA}$)  and receives payoff ${r}^{\sD}(s_t, d_t)$ and ${r}^{\sA}(s_t, a_t)$, respectively, and the game transitions to a next state $s_{t+1}$. This is continued until they reach an absorbing state  and incur ${c}^{\sA}(s_t)$ and ${c}^{\sD}(s_t)$, respectively, or the game arrives at the horizon, i.e., $t=T$.  Then, 
 
\begin{equation}\label{eq:new-1}
r^{\sA}(s_t,  a_t)=0~\mbox{ for ~all ~} s_t, a_t,
\end{equation}

\begin{eqnarray}\label{eq:new-2}
{c}^{\sA}({s}_t)\hspace*{-2 mm}&=&\hspace*{-2 mm}
\begin{cases}
\begin{array}{ll}
 \alpha_{\sA},&  s_t \in {\bf S}_{\sA}, ~ s_t= (A, \tau,  \ldots, v_{i_W}) \\
\beta_{\sA},&  s_t \in {\bf S}_{\sA}, ~ {s}_t= (A, v_{i_1},  \ldots, \tau,  \ldots, v_{i_W}) \\
0,& {\rm otherwise}
\end{array}
\end{cases}
\end{eqnarray}
\begin{eqnarray}\label{eq:new-3}
{c}^{\sD}({s}_t)\hspace*{-2 mm}&=&\hspace*{-2 mm}
\begin{cases}
\begin{array}{ll}
 \beta_{\sD},&  s_t \in {\bf S}_{\sD}, ~{s}_t=  (D, v_{i_1}, \ldots, v_{i_W}),  v_{i_1} \in  \D \\
  \alpha_{\sD},&  s_t \in {\bf S}_{\sD}, ~ s_t= (A, \tau,  \ldots, v_{i_W}) \\
0,& {\rm otherwise}
\end{array}
\end{cases}
\end{eqnarray}
\begin{eqnarray}\label{eq:new-4}
{r}^{\sD}(s_t, d_t)\hspace*{-2 mm}&=&\hspace*{-2 mm}
\begin{cases}
\begin{array}{ll}
 \hspace*{-2 mm}0,& s_t \in {\bf S}_{\sD}, ~d_t=0\\
\hspace*{-2 mm}  \C_{\sD}(v_{i_k}),& s_t \in {\bf S}_{\sD}, ~d_t=v_{i_k}, k \in \{1, \ldots, W\}.
\end{array}
\end{cases}
\end{eqnarray}

As the initial state of the game is ${\bf s}_0$, for a strategy pair $({p}_{\sD}, {p}_{\sA})$ the expected payoffs of the players are 
\begin{eqnarray}
{U}_{\sA}({p}_{\sD}, {p}_{\sA}) &=& \mathbb{E}_{{\bf s}_0,{p}_{\sA}, {p}_{\sD}} \left[\sum\limits_{t = 0}^{T}({R}^{\sA}_t)\right] \mbox{and}\label{eq:Adv_obj1}\\
{U}_{\sD}( {p}_{\sD}, {p}_{\sA}) &=& \mathbb{E}_{{\bf s}_0,{p}_{\sA}, {p}_{\sD}} \left[\sum\limits_{t = 0}^{T}({R}^{\sD}_t)\right],\label{eq:Def_obj1}
\end{eqnarray}
where $\mathbb{E}_{{\bf s}_0,{p}_{\sA}, {p}_{\sD}}$ denotes the expectation with respect to ${\bf s}_0, {p}_{\sA}$, and ${p}_{\sD}$ and from Eqs.~\eqref{eq:new-1}-\eqref{eq:new-4}
\begin{equation}\label{eq:new-5}
{R}^x_t=
\begin{cases}
\begin{array}{ll}
{r}^{\sD}(s_t, d_t),& x=D,~ s_t \in {\bf S}_{\sD},~t < T\\
{r}^{\sA}(s_t, a_t),& x=A,~ s_t \in {\bf S}_{\sA},~t < T\\
{c}^{\sD}(s_t),& s_t \in {\bf S}_{\sD},~t=T,\\
{c}^{\sA}(s_t),& s_t \in {\bf S}_{\sA},~t=T.
\end{array}
\end{cases}
\end{equation}

We incorporate the idea of maximizing the probability of detection and minimizing the  probability of adversary evading  detection using the terms ${p}_{\sT}$ and ${p}_{\sR}$, respectively.  Further, we capture the false-positives generated in the system using the term ${p}_{\sFP}$. Note that, generation of false-positive implies that the defender failed to capture the adversary and concluded a benign flow is malicious. This means that the adversary will attain the target as the malicious flow evaded detection. 
Given a set of strategies $({p}_{\sD} , {p}_{\sA} )$, where ${p}_{\sD} \in {{\bf P}}_{\sD}$ and ${p}_{\sA} \in {{\bf P}}_{\sA}$, and benign distribution $\pi_{\sB}$, the payoff functions of the players, i.e., Eqs.\eqref{eq:Adv_obj1}, \eqref{eq:Def_obj1}, can be rewritten using Eqs.~\eqref{eq:new-1}-\eqref{eq:new-4}, \eqref{eq:new-5} as
{\scalefont{0.95}{
\begin{eqnarray}
{U}_{\sD}({p}_{\sD} , {p}_{\sA} ) &\hspace*{-3 mm}=&\hspace*{-3 mm}{p}_{\sT}\, \alpha_{\sD}+ ({p}_{\sR}+{p}_{\sFP})\, \beta_{\sD} \hspace*{-1 mm}+ \hspace*{-2 mm}\sum_{{\bf s} \in \bhS_{\sD}}   \sum_{v_i \in {\bf s} }\Big({p}_{\sD}(v_i) \C_{\sD}(v_i)\Big)\hspace*{-1 mm}, \label{eq:Ud1}\\
{U}_{\sA}({p}_{\sD} , {p}_{\sA} )&\hspace*{-3mm}=& \hspace*{-4mm}   \Big({p}_{\sT}\, \alpha_{\sA}+ ({p}_{\sR}+{p}_{\sFP})\, \beta_{\sA} \Big). \label{eq:Ua1}
\end{eqnarray}
}}
Here  ${p}_{\sT}$ is the cumulative probability that the  adversarial flow is detected by the defender. The term ${p}_{\sT}\, \alpha_{\sD}$ captures the  reward received by DIFT for detecting the attack and the term ${p}_{\sT}\, \alpha_{\sA}$ captures the penalty incurred by APT for getting detected.  Also, ${p}_{\sR}$ denotes the cumulative probability that the adversarial flow  reaches a destination and ${p}_{\sFP}$ denotes the cumulative probability that a benign flow is concluded as malicious (i.e., trapped) by the defender, i.e., false-positive. The term $({p}_{\sR}+{p}_{\sFP})\, \beta_{\sD}$  captures the  penalty incurred by DIFT for not detecting the attack and the term $({p}_{\sR}+{p}_{\sFP})\, \beta_{\sA}$ captures the  reward received by the APT for  not getting detected. Recall that ${p}_{\sD}(v_i)$ denotes the probability with which DIFT selects node $v_i$ as a security check point (trap). The term $\sum_{{\bf s} \in \bhS_{\sD}}   \sum_{v_i \in {\bf s} }({p}_{\sD}(v_i) \C_{\sD}(v_i))$ captures the total resource cost associated with DIFT for performing  security analysis.  Note   that ${p}_{\sT}$, ${p}_{\sR}$ are functions of ${p}_{\sD}$ and ${p}_{\sA}$, and ${p}_{\sFP}$  is a function of ${p}_{\sD}$  and $\pi_{\sB}$.  Our focus is to compute a  limiting average equilibrium of the nonzero-sum  stochastic game.
\section{Computation of Optimal Strategy}\label{sec:comp}
\subsection{Solution Concept}\label{sec:solution_concept}
In this subsection, we present the solution concept of the game. 
\begin{defn}\label{def:BR}
Let ${p}_{\sA}: {\Z}^{\*}  \rightarrow [0, 1]^{|\A_{\sA}|}$ denote a strategy of the  adversary. Also let ${p}_{\sD}: {\Y}^{\*}   \rightarrow [0, 1]^{|\A_{\sD}|}$ denote a strategy of the  defender, i.e., the probability of performing security analysis at a state. Then the best response of the defender is given by
$$\mbox{BR}({{p}_{\sA}) = \arg\max_{{p}_{\sD} \in {{\bf P}}_{\sD}}}{\{{U}{\sD}({{{p}}}_{\sD}, {{{p}}}_{\sA})\}}.$$ 
Similarly, the best responses of the adversary  are given by $$\mbox{BR}({{{p}}}_{\sD}) = \arg\max_{{p}_{\sA}\in {{\bf P}}_{\sA}}{\{{U}{\sA}({{p}}_{\sD}, {{{p}}}_{\sA}) \}}.$$
\end{defn}
The best response of the defender is the set of defense strategies that maximize the payoff of the defender for a given adversarial strategy and known benign flow distribution. The best response of the adversary is a set of transition strategies, for a given defender strategy and known benign flow distribution, that maximizes the probability of adversary reaching a destination node without getting detected.
\begin{defn}\label{def:NE}
A pair of strategies $({p}_{\sD}, {p}_{\sA})$ is said to be a Nash equilibrium (NE)  if $$ {p}_{\sD} \in \mbox{BR}({p}_{\sA}) \mbox{~and~} {p}_{\sA} \in \mbox{BR}({p}_{\sD}).$$
\end{defn}
An NE is a pair of strategies such that no player can benefit through unilateral deviation, i.e., by changing the strategy while the other player keeps the strategy unchanged.

  \begin{lemma}\label{lem:cardinality_SS}
 Consider the  APT vs. DIFT game ${\bf G} = \{\bS, {\bf s}_0, \Sigma, \Delta, P \}$. Let $N$ be the number of nodes in the IFG and $W$ be the number of flows that arrive into the system at time $t=1$. Then $|\bS| = O(N^W)$.
 \end{lemma}
 \begin{proof}
We prove the result by showing that $|{\bf S}_{\sA}| = O(N^W)$ and $|{\bf S}_{\sD}| = O(N^W)$.   Consider an  arbitrary state ${\bf s} = (x, v_{i_1}, \ldots, v_{i_W} ) \in {\bf S} \setminus {\bf s}_0$, where $x \in \{A, D\}$. Here, $v_{i_k}\in V_{\G} \cup \{\phi, \tau\}$ for all $k \in \{1, \ldots, W\}$.  All the $W$ tagged flows arrive at distinct nodes in the set $V_{\G}$ of the IFG as a process or file in the system receive exactly one information flow at a particular time. Therefore, for a state  ${\bf s} =(x, v_{i_1}, \ldots, v_{i_W})$, where  $v_{i_k}\in V_{\G}$ for all $k \in \{1, \ldots, W\}$, $i_a \neq i_b$ for $a, b \in \{1, \ldots, W\}$.  Also, the defender performs security analysis on one tagged flow at a time and hence there is no state ${\bf s} =(x, v_{i_1}, \ldots, v_{i_W})$ with $v_{i_a}=v_{i_b}=\tau$ for $a\neq b$.   
 Repetition of $v_{i_k}$'s is possible only for states with $\phi$. Hence ${\bf s}$ belongs to one of the two types: (a)~$\{{\bf s} =(x, v_{i_1}, \ldots, v_{i_W}): v_{i_k}\in V_{\G} \cup \{\tau\}\}$ and (b)~$\{{\bf s} =(x, v_{i_1}, \ldots, v_{i_W}): v_{i_k}\in V_{\G} \cup \{\phi\}$ such that $v_{i_k}=\phi$ for at least some $k \in \{1, \ldots, W\}\}$.


Case~(a) corresponds to selecting $W$ items from a set of $N+1$ items without any repetitions. The cardinality of this set is ${N+1}\choose {W}$ $ = \frac{(N+1)!}{(W)!(N+1-W)!} = O(N^W)$. Case~(b) corresponds to states with one or more $\phi$ entries. Note that, here repetitions are allowed only for  $\phi$. The cardinality of this set is ${{N+1}\choose {W-1}}+ {{N+1}\choose {W-2}}+ \ldots+ {{N+1}\choose {2}} +  {{N+1}\choose {1}} = O(N^{W-2})$. Additionally, there is an initial state ${\bf s}_0$. From (a) and~(b), the number of possible cases for ${\bf s}$ is $O(N^W)$. Since ${\bf s} = (x, v_{i_1}, \ldots, v_{i_W} )$, where $x \in \{A, D\}$, we get  $|\bS_{\sA}| = O(N^W)$ and $|\bS_{\sD}| = O(N^W)$. As $\bS = \bS_{\sA} \cup \bS_{\sD}$, we get $|\bS| =  O(N^W)$.
 \end{proof}
 
It is shown that there exists an NE for nonzero-sum discounted  stochastic games \cite{sobel1971noncooperative}. However, the existence of NE for nonzero-sum undiscounted stochastic games is open when the time horizon is infinite, i.e., $T=\infty$. In the result below  we prove the existence of NE for the game ${\bf G}$. The existence is shown by proving that time horizon is indeed finite for ${\bf G}$ and then invoking the existence of NE for finite horizon undiscounted games.

\begin{prop}\label{prop:TH}
Let $T$ denote the termination time of the APT vs. DIFT game. The APT vs. DIFT game, ${\bf G}$, terminates in at most  is $2N$ steps, i.e., $T \leqslant 2N$.
\end{prop}
\begin{proof}

The action set $\A_{\sA}$ of the adversary player (APT) is to choose a transition along an attack path  of the flow-network $\F$ (Definition~\ref{def:attack_path})  or to drop out at some point. Let $\Omega_{\D}$ denote the set of attack paths in $\F$.   Consider an arbitrary attack path $\hat{\omega} \in \Omega_{\D}$. Recall that $\hat{\omega}$ is a simple path as an   APT  due to its stealthy behavior will not traverse in cycles through the system. Hence the length of $\hat{\omega}$ is at most $N$, as $\hat{\omega}$ is a simple directed path. Thus in any run of the game, the  adversary can take at most $N$ transitions. As ${\bf G}$ is a turn-based game,  the game terminates in at most $2N$ steps.
\end{proof}


\begin{prop}\label{prop:NE}
There exists a Nash equilibrium for the APT vs. DIFT game ${\bf G}$.
\end{prop}
\begin{proof}
It is shown in  \cite{HesPra-01} that there exists a Nash equilibrium for a nonzero-sum stochastic game with asymmetric information structure, under stochastic behavioral strategies, when the time horizon is finite. Proposition~\ref{prop:TH} prove that the APT vs. DIFT game will terminate in $2N$ number of steps. Further, the behavioral strategy space is a subset of the strategy space ${\bf P}_{\sA} \times {\bf P}_{\sD}$. Hence by the result in  \cite{HesPra-01}, the proof follows.
\end{proof}

\subsection{Solution Approach}\label{sec:solution_approach}

In this section, we compute  an NE of the DIFT vs. APT game ${\bf G}$.
 Our approach is based on a minimum capacity cut-set formulation on a flow-network constructed from the information flow graph of the system.
Consider the flow-network $\F=(V_{\F}, E_{\F})$, where $V_{\F} = V_{\G} \cup \{s_{\sF}, t_{\sF}\}$ and $E_{\F} = E_{\G} \cup \{s_{\sF}\times \lambda\} \cup \{\D \times t_{\sF}\}$. 
Then, a {\em cut} of $\F$ is defined below.
\begin{defn}\label{def:cut}
In a flow-network $\F$ with vertex set $V_{\sF}$ and directed edge set  $E_{\sF}$,  the cut induced by $\hat{\S} \subset V_{\sF}$  is a subset of edges $\kappa(\hat{\S}) \subseteq E_{\sF}$ such that for every $(u,v) \in \kappa(\hat{\S})$, $|\{u,v\} \cap \hat{\S}| = 1$. Further, given edge capacity vector $c_{\sF}: E_{\sF} \rightarrow \R_{+}$, the capacity of a cut $\kappa(\hat{\S})$, is defined  as the sum of the capacities of the edges in the cut, i.e., $c_{\sF}(\kappa(\hat{\S})) = \sum_{e \in \kappa(\hat{\S})} c_{\sF}(e).$
\end{defn}

The  {\em (source-sink)-min-cut problem} aims to find a cut $\kappa(\hat{\S}^\*)$  of  $\hat{\S}^\* \subset V_{\sF}$ such that $c_{\sF}(\kappa(\hat{\S}^\*)) \leqslant c_{\sF}(\kappa(\hat{\S}))$
for any cut  $\kappa(\hat{\S})$ of $\hat{\S} \subset V_{\sF}$ satisfying $s_{\sF} \in \hat{\S}$ and $t_{\sF} \notin \hat{\S}$. The (source-sink)-min-cut problem is  well studied and there exist algorithms to find a min-cut in time  polynomial   in $|V_{\sF}|$ and $|E_{\sF}|$ \cite{Orl:93}. In our approach to compute NE, which is detailed later in the section,  we find a min-cut through nodes of the flow-network instead of the edges. Hence we transform the node version of the min-cut problem  to an equivalent edge version. For this, we introduce an edge corresponding to each node in $\F$ except the source and sink nodes.  
The transformed flow-network is denoted as $\overbar{\F} = (\overbar{V}_{\sF}, \overbar{E}_{\sF})$, where $ \overbar{V}_{\sF} = V_{\F} \cup V'_{\G} $ and $ \overbar{E}_{\sF} = E'_{\F} \cup E'_{\G} \cup E_{\lambda}\cup E_{\sD}$ with $V'_{\G} = \{v'_1, \ldots, v'_N\}$, ${E'}_{\F}=\{(v'_i, v_j): (v_i, v_j)\in E_{\G}\}$, $E_{\lambda} = \{(s_{\sF}, v_i): v_i \in \lambda\}$,  $E_{\sD} = \{(v'_i, t_{\sF}): v_i \in \D\}$, and   $E'_{\G}=\{(v_i, v'_i): i=1,\ldots, N\}$. The capacity vector associated with the edges in $\overbar{\F}$ is given by
\begin{equation}\label{eq:capa}
 c_{\sF} (e) := \begin{cases}
\begin{array}{ll}
\C_{\sD}(v_i), & e \in E'_{\G} \\
\infty, & \mbox{otherwise}
\end{array}
\end{cases}
\end{equation}

 Note that $E'_{\G}$ is a cut of $\overbar{\F}$ and $\sum_{e \in E'_{\G}}c_{\sF}(e) < \infty$. Hence any minimum capacity cut in $\overbar{\F}$ corresponds to edges from the set $E'_{\G}$ as the capacity of the remaining edges  is $\infty$ as shown  in Eq.~\eqref{eq:capa}. Further, the capacity of an edge in set $(v_i, v'_i) \in E'_{\G}$ corresponds to the cost of conducting the security analysis at  node $v_i$. Thus a minimum capacity cut in $\overbar{\F}$ corresponds to a cut node set of the IFG with minimum total cost of performing security analysis.

 Let  $\kappa(\hat{\S}^\*)$ denotes  an optimal solution to the (source-sink)-min-cut problem on $\overbar{\F}$.  Then $\kappa(\hat{\S}^\*) \subseteq E'_{\G}$ and $c_{\sF}(\kappa(\hat{\S}^\*))<\infty$.   The nodes corresponding to the min-cut is 

\begin{equation}\label{eq:mincut}
 \hat{\S}^\* := \{v_i: (v_i, v'_i) \in \kappa(\hat{\S}^\*)\}.
 \end{equation}

In the APT vs. DIFT game, the aim of the defender is to  optimally select trapping nodes in the IFG, i.e., nodes of IFG to perform security analysis,  such that  no adversarial flow reaches  some node in $\D$. In other words, defender ensures that all adversarial flows that originate in node  $s_{\sF}$ gets detected before reaching node $t_{\sF}$. In order to ensure  security, the defender must  select at least  one node in all possible paths from $s_{\sF}$ to $t_{\sF}$ as a trapping node. In the equilibrium analysis of the game we prove that an optimal strategy of the defender is indeed to select the min-cut nodes  of the flow-network as trapping nodes. 
The objective of the adversary is to optimally choose the transitions in such a way that the probability of reaching $t_{\sF}$ is maximum.  The adversary hence plans its transitions to select an {\em attack path} with least probability of detection.  

The result below proves  that an NE of game ${\bf G}$ is represented by the nodes corresponding to a minimum capacity cut in the flow-network $\overbar{\F}$. Note that, the solution to the min-cut problem may not be unique.  Consequently, there may exist multiple NE for the game ${\bf G}$. 
\begin{theorem}\label{thm:mincutNE}
Every NE $(p_{\sA}, p_{\sD})$ of the APT vs. DIFT game ${\bf G}$ satisfies the following properties:
\begin{itemize}
\item[1)] The defender's strategy $p_{\sD}$ selects all the nodes in $ \hat{\S}^\*$ as trapping nodes, where $ \hat{\S}^\*$ is a set of min-cut nodes of the flow network $\overbar{\F}$. 
\item[2)] The adversary's strategy $p_{\sA}$ chooses transitions such that each attack path passes through exactly one node in the set $ \hat{\S}^\*$.
\end{itemize}
\end{theorem}

We prove Theorem~\ref{thm:mincutNE}  by invoking the property that at NE every player plays a best response against the other players simultaneously. We first present prove the best response results, Lemma~\ref{lem:NEdisj}, Lemma~\ref{lem:NEBR}, and Lemma~\ref{lem:detection_equal}, and then present the proof of  Theorem~\ref{thm:mincutNE}.

Lemma~\ref{lem:NEdisj} gives the best response of the adversary for a given defender strategy that selects the min-cut nodes of $\overbar{\F}$ as the trapping nodes for analyzing the flows.
\begin{lemma}\label{lem:NEdisj}
Let $\Omega_{\D}$ be the set of all attack paths in $\F$. Consider a defender strategy $p_{\sD}$ in which only the min-cut nodes  $ \hat{\S}^\*$ of the flow-network $\overbar{\F}$ are chosen as  trapping nodes with nonzero probability.  Then, the best response of the adversary, $BR(p_{\sD})$, is  to choose the transitions in such a way that all attack paths which has nonzero probability under $BR(p_{\sD})$  pass through exactly one node in $\hat{\S}^\*$.
\end{lemma}
\begin{proof}
Consider the payoff function of the adversary, $U_{\sA}(p_{\sD}, p_{\sA}) =(p_{\sT}\, \alpha_{\sA}+ (p_{\sR}+p_{\sFP})\, \beta_{\sA} )$. Here, $p_{\sT}$ and $p_{\sR}$ are functions of $p_{\sD}$ and $p_{\sA}$. However, $p_{\sFP}$ depends only on $p_{\sD}$ and $\pi_{\sB}$. Thus for a given $p_{\sD}, \pi_{\sB}$, the probabilities $p_{\sT}$ and $p_{\sR}$ vary depending on $p_{\sA}$, however, $p_{\sFP}$ is a constant.  We prove the result using a contradiction argument. Suppose the best response of $\P_{\sA}$ is a strategy $p'_{\sA}$ such that  there exists  a path $\hat{\omega} \in \Omega_{\D}$ that passes through two nodes, say $v_i, v_j \in \hat{\S}^\* \subset V_{\G}$, and $\pi_{\sA}( \hat{\omega} )\neq 0$, where $\pi_{\sA}( \hat{\omega} )$ is the probability of attack path $ \hat{\omega} $ under $p'_{\sA}$. Note that $\hat{\S}^\*$ is a min-cut and  $v_i, v_j \in \hat{\S}^\*$. Further,  $\C_{\sD}(v_i)\neq 0$ and $\C_{\sD}(v_j)\neq 0$. Thus there exists at least one directed path, say $ \hat{\omega}'$, from $v_i$ to $t_{\sF}$ that does not pass through any node in $\hat{\S}^\* \setminus \{v_i\}$. Similarly, there exists at least one directed path, say $ \hat{\omega}''$, from $v_j$ to $t_{\sF}$ that does not pass through any node in $\hat{\S}^\* \setminus \{v_j\}$. As per the given defender strategy $p_{\sD}$ the only trapping node in $ \hat{\omega}',  \hat{\omega}''$ is $v_i, v_j$, respectively. Thus there exists a strategy $p_{\sA}$ such that all paths in $\Omega_{\D}$ with nonzero probability under $p_{\sA}$ pass through exactly one node in  $\hat{\S}^\*$ and  $U_{\sA}(p_{\sD}, p_{\sA})>U_{\sA}(p_{\sD}, p'_{\sA})$ (since $p_{\sFP}$ remains same and $p_{\sT}$ is higher and $p_{\sR}$ is lower under $p'_{\sA}$ when compared to the values under strategy $p_{\sA}$). This contradicts the assumption that $p'_{\sA}$  is a best response  and   completes the proof.
\end{proof}

Recall that $ \Omega_{\pazocal D}$ is the set of all source-to-sink paths in the flow-network $\pazocal{F}$ (Definition~IV.1). Note that, any path in  $\pazocal{F}$ that does not belong to $ \Omega_{\pazocal D}$ is not a valid attack, as the attacker cannot reach a destination node.  Every  attack path in $ \Omega_{\pazocal D}$, depending on the defender strategy and the transition structure of the game, induces a  set of paths  in the state space graph of the APT vs. DIFT game. Let $\Omega$ be the set of all paths induced by  $ \Omega_{\pazocal D}$.  In other words, $\Omega$ is the set of all possible state transition paths in the state space graph corresponding to all attack paths in $\pazocal{F}$.

\begin{defn}\label{def:induced_path}
  Let the set of paths induced  in the state space $\bS$ by the attack paths $\Omega_{\D}$ in $\F$ be denoted as $\Omega$. Then the probability of selecting a path $\omega \in \Omega$, denoted by $\pi(\omega)$,  is $\pi(\omega) = \pi_{\sA}(\omega)\pi_{\sB}(\omega)$, where $\pi_{\sA}(\omega)$ is the probability with which an adversary chooses $\omega$ (product of the adversary transition probabilities along path $\omega$) and $\pi_{\sB}(\omega)$ is the probability of benign flows in $\omega$ under distribution $\pi_{\sB}$.
\end{defn}
  
The following result proves  that, under certain conditions, for a given adversary strategy  the best response of the defender is to select one node in every  attack path as a trapping node. 

\begin{lemma}\label{lem:NEBR}
Let $\Omega_{\D}$ be the set of attack paths in $\F$, $\Omega$ be the set of paths in $\bS$ induced by $\Omega_{\D}$, and $p_{\sA}$ be a given strategy of $\P_{\sA}$. For $\omega \in \Omega$,  let $\omega(A)$ denote the set of nodes in $\omega$ corresponding to the adversarial (malicious) flow and $\omega(B)$ denote the set of nodes in $\omega$ corresponding to the benign flows. Also, let $p(\omega)$ denote the probability of detecting the adversary along path $\omega$, i.e., $p(\omega) = \Big[1-\prod_{v_i \in \omega(A)} (1-p_{\sD}(v_i))\Big](1-FN)$. Also, let $f(\omega)$ denote the probability of trapping a benign flow (false-positive), i.e., $f(\omega) = \Big[1-\prod_{v_i \in \omega(B)} (1-p_{\sD}(v_i))\Big]FP$.   If the defender's strategy satisfies the following two conditions:
\begin{enumerate}
\item[(a)] $p(\omega) = p(\omega') $, for all  $\omega, \omega' \in \Omega$ and
\item[(b)] $f(\omega) = f(\omega') $, for all $\omega, \omega' \in \Omega$, 
\end{enumerate}
then the best response of the defender, $BR(p_{\sA})$,  is to select with nonzero probability exactly one node in every $\hat{\omega} \in \Omega_{\D}$ as a trapping node. 
\end{lemma}
\begin{proof}
Let $\pi(\omega)$ denote the probability of a path $\omega \in \Omega$ under the given strategy $p_{\sA}$ and benign distribution $\pi_{\sB}$. For a path $\omega$ in the state space $\bS$,  $\omega(A)$ denotes the set of nodes in $\omega$ corresponding to the adversarial (malicious) flow and $\omega(B)$ denotes the set of nodes in $\omega$ corresponding to the benign flows.  Let $p(\omega)$ denote the probability of detecting the adversary along path $\omega$, i.e., $p(\omega) = \Big[1-\prod_{v_i \in \omega(A)} (1-p_{\sD}(v_i))\Big](1-FN)$. Also, let $f(\omega)$ denote the probability of trapping a benign flow, i.e., false-positive.  Then $f(\omega) = \Big[1-\prod_{v_i \in \omega(B)} (1-p_{\sD}(v_i))\Big]FP$. The defender's payoff 
\begin{eqnarray*}
U_{\sD}(p_{\sD}, p_{\sA}) &=& \sum_{\omega \in \Omega}\pi(\omega)\Big[p(\omega)\,\alpha_{\sD} + (1-p(\omega) +f(\omega))\beta_{\sD} \Big]\\
&+&\sum_{ \omega \in \Omega}\Big(  \sum_{v_i \in \omega}p_{\sD}(v_i)\C_{\sD}(v_i)\Big)
\end{eqnarray*}
Given $p(\omega)$'s and $f(\omega)$'s are equal for all $\omega \in \Omega$. Thus 
\begin{eqnarray*}
U_{\sD}(p_{\sD}, p_{\sA})&=& \hspace*{-3 mm}\Big[p(\omega)\,(\alpha_{\sD}-\beta_{\sD}) + (1+f(\omega))\beta_{\sD} \Big] \Big(\sum_{\omega \in \Omega}\pi(\omega)\Big)\nonumber\\
&+&  \sum_{ \omega \in \Omega}\Big(  \sum_{v_i \in \omega}p_{\sD}(v_i)\C_{\sD}(v_i)\Big)
\end{eqnarray*}
\begin{equation}\label{eq:add}
= \Big[p(\omega)\,\alpha_{\sD} + (1-p(\omega)+f(\omega))\beta_{\sD} \Big] + \sum_{ \omega \in \Omega}\Big(  \sum_{v_i \in \omega}p_{\sD}(v_i)\C_{\sD}(v_i)\Big)
\end{equation}
Eq.~\eqref{eq:add} holds as $\sum_{\omega \in \Omega}\pi(\omega)=1$.
 Consider a defender strategy $ p_{\sD}$  in which exactly one node in every $\hat{\omega} \in \Omega_{\D}$ is chosen as the trapping node. Assume that the defender strategy is modified to $p'_{\sD}$ such that more than one node in some path are chosen as trapping nodes. This variation updates the probabilities of nodes in a set of paths in $\Omega$. 
Note that, due to the constraints on $p(\omega)$ and $f(\omega)$ (conditions~(a) and~(b)), the defender's probabilities (strategy) at two nodes in  a path are dependent. Hence for all paths $\omega\in \Omega$ whose probabilities are modified in $p'_{\sD}$,  $\sum_{v_i \in \omega} p'_{\sD}(v_i)\C_{\sD}(v_i) < \sum_{v_i \in \omega} p_{\sD}(v_i)\C_{\sD}(v_i)$. This holds as the probability in the single node case is less than the sum of the probabilities of more than one node case as the events are dependent and the sum of the values of $\C_{\sD}(\cdot)<0$  are also minimum (since min-cut). This implies 
\begin{equation}\label{eq:BR_D1}
\sum_{ \omega \in \Omega}\Big( \sum_{v_i\in V_{\G}}p'_{\sD}(v_i)\C_{\sD}(v_i) \Big) < \sum_{ \omega \in \Omega}\Big( \sum_{v_i\in V_{\G}}p_{\sD}(v_i)\C_{\sD}(v_i) \Big). 
\end{equation}
From Eq.~\eqref{eq:add} and by conditions~(a) and~(b), we get   $U_{\sD}(p'_{\sD}, p_{\sA}) < U_{\sD}(p_{\sD}, p_{\sA})$. Therefore, $p'_{\sD}$ is not a best response for the defender and  no best response of the defender has more than one node chosen as trapping node, if $p(\omega)$'s and $f(\omega)$'s  are equal for all $\omega \in \Omega$. 
\end{proof}

The result below proves that if Lemma~\ref{lem:NEBR} holds, then for a given adversary strategy the  best response of the defender is to select the min-cut nodes of $\overbar{\F}$ as trapping nodes.
\begin{lemma}\label{lem:detection_equal}
Let $\Omega_{\D}$ be the set of attack paths in $\F$. Assume that the defender's strategy satisfies conditions~(a) and~(b) in  Lemma~\ref{lem:NEBR} for a given adversary strategy $p_{\sA}$. Then  the best response of the defender, $BR(p_{\sA})$,  is to select with nonzero probability a  min-cut node set of the flow-network $\overbar{\F}$ as the trapping nodes.
\end{lemma}
\begin{proof}
Under conditions~(a) and~(b) in Lemma~\ref{lem:NEBR}, the best response of the defender is to select   one node in every attack path as trapping node. Note that all attack paths in $\F$ pass through some node in the min-cut node set $\hat{\S}^\*$. By selecting the nodes in  $\hat{\S}^\*$ as trapping nodes, all possible attack paths have some nonzero probability of getting detected. We  prove  the result using a contradiction argument. Suppose that the best response of the defender is not to select the nodes in $\hat{\S}^\*$ as trapping nodes. Then, there exists a subset of nodes $\hat{\S} \subset V_{\G}$ such that $\sum_{v_i \in \hat{\S}} \C_{\sD}(v_i) < \sum_{v_j \in \hat{\S}^\*} \C_{\sD}(v_j) $. Further, all possible attack paths in $\F$ pass through some node in $\hat{\S}$. Then,  $\hat{\S}$ is a (source-sink)-cut-set and let $\kappa(\hat{\S}) := \{(v_i, v'_i): v_i \in \hat{\S}\}$. Then, $\kappa(\hat{\S})$ is a cut set and $c_{\sF}(\kappa(\hat{\S})) < c_{\sF} (\kappa(\hat{\S}^\*))$. This contradicts the fact that $\kappa(\hat{\S}^\*)$ is an optimal solution to the (source-sink)-min-cut problem.  Hence under Lemma~\ref{lem:NEBR} the best response of the defender is to select the nodes in $\hat{\S}^\*$ as trapping nodes. 
\end{proof}

Using Lemma~\ref{lem:NEdisj}, Lemma~\ref{lem:NEBR}, and Lemma~\ref{lem:detection_equal},  we present below the proof of Theorem~\ref{thm:mincutNE}.

\noindent{\em Proof~of~Theorem~\ref{thm:mincutNE}}: Consider a defender's strategy that selects the min-cut nodes as the trapping nodes. Then, by Lemma~\ref{lem:NEdisj} the best response of the adversary is to select transitions in such a way that all attack paths with nonzero probability pass  through exactly one node in the min-cut. Further
 Let $\Omega_{\D}$ be the set of attack paths in $\F$ and $\Omega$ be the set of paths in $\bS$ induced by $\Omega_{\D}$. Lemma~\ref{lem:detection_equal}  concludes that the best response of the defender is to select the  min-cut nodes of  $\F$ as trapping nodes, provided the probability of detecting the adversary and the probability of trapping a benign flow are equal for all $\omega \in \Omega$.  This implies that, if NE strategy pair satisfy the conditions that $p(\omega)$'s and $f(\omega)$'s  are equal for all $\omega \in \Omega$, then  the defender's strategy at NE is to select the min-cut nodes as the trapping nodes  and the adversary's strategy is to choose an attack path such that it passes through exactly one node  in the min-cut node set.  By Proposition~\ref{prop:NE}, there exists an NE for ${\bf G}$. Consequently, if $p(\omega)$'s and $f(\omega)$'s are equal at NE, the proof follows.
 
 Let  $(p_{\sA}, p_{\sD})$ be an NE of ${\bf G}$.  Consider a unilateral deviation in the strategy of the adversary. Let $\pi(\omega)$'s for $\omega \in \Omega$ are modified due to change in transition probabilities of the adversary such that the  updated  probabilities of the attack paths are $\pi(\omega_i)+\epsilon_i$, for $i=1, \ldots, |\Omega|$. Here, $\epsilon_i$'s can take positive values, negative values or zero such that $\sum_{i=1}^{|\Omega|}\epsilon_i = 0$.  Consider two arbitrary paths, say $\omega_1 \in \Omega$ and $\omega_2 \in \Omega$, such that a unilateral change in the adversary strategy changes $\pi(\omega_1)$ and $\pi(\omega_2)$ and the probabilities of the other paths remain unchanged. Without loss of generality, assume that $\pi(\omega_1)$ increases by $\epsilon$ while   $\pi(\omega_2)$ decreases by $\epsilon$ and all other $\pi(\omega)$'s remain the  same. As $(p_{\sD}, p_{\sA})$ is an NE, $(\pi(\omega_1) + \epsilon)\Big( p(\omega_1)(\alpha_{\sA}-\beta_{\sA}) + \beta_{\sA} +f(\omega_1)\,\beta_{\sA}\Big) + (\pi(\omega_2) - \epsilon)\Big( p(\omega_2)(\alpha_{\sA}-\beta_{\sA}) + \beta_{\sA} +f(\omega_2)\,\beta_{\sA}\Big) \leqslant \pi(\omega_1) \Big( p(\omega_1)(\alpha_{\sA}-\beta_{\sA}) + \beta_{\sA} +f(\omega_1)\,\beta_{\sA}\Big) + \pi(\omega_2)\Big( p(\omega_2)(\alpha_{\sA}-\beta_{\sA}) + \beta_{\sA} +f(\omega_2)\,\beta_{\sA}\Big)$. Thus 
\begin{equation}\label{eq:one}
(p(\omega_1)-p(\omega_2))(\alpha_{\sA}-\beta_{\sA}) +(f(\omega_1)-f(\omega_2))\beta_{\sA} \leqslant 0.
\end{equation}
Now consider another unilateral deviation of adversary strategy such that $\pi(\omega_1)$ decreases by $\epsilon$ while   $\pi(\omega_2)$ increases by $\epsilon$ and all other $\pi(\omega)$'s remain the  same. This gives 
\begin{equation}\label{eq:two}
(p(\omega_1)-p(\omega_2))(\alpha_{\sA}-\beta_{\sA}) +(f(\omega_1)-f(\omega_2))\beta_{\sA} \geqslant 0.
\end{equation}
Eqs.~\eqref{eq:one} and~\eqref{eq:two} imply  
\begin{equation}\label{eq:three}
(p(\omega_1)-p(\omega_2))(\alpha_{\sA}-\beta_{\sA}) +(f(\omega_1)-f(\omega_2))\beta_{\sA} = 0.
\end{equation}
Note that $(\alpha^{\sA} - \beta^{\sA}) <0$ and $\beta_{\sA} >0$. Therefore, there are three possible cases where Eq.~\eqref{eq:three} holds. 
\begin{enumerate}
\item[(i)] $(p(\omega_1)-p(\omega_2)) > 0$ and $(f(\omega_1)-f(\omega_2)) > 0$,
\item[(ii)] $(p(\omega_1)-p(\omega_2)) < 0$ and $(f(\omega_1)-f(\omega_2)) < 0$, and
\item[(iii)] $(p(\omega_1)-p(\omega_2)) = 0$ and $(f(\omega_1)-f(\omega_2))= 0$.
\end{enumerate}
Consider case~(i). For a path $\omega$ in the state space  let $\omega(A)$ denote the set of nodes in $\omega$ corresponding to the adversarial flow and $\omega(B)$ denote the set of nodes in $\omega$ corresponding to the benign flows.  Rewriting$(p(\omega_1)-p(\omega_2)) > 0$, we get 
{\scalefont{0.9}{ 
\begin{equation}\label{eq:case_one_1}
\Big(\Big[1-\prod_{v_i \in \omega_1(A)} (1-p_{\sD}(v_i))\Big]- \Big[1-\prod_{v_j \in \omega_2(A)} (1-p_{\sD}(v_j))\Big]\Big)(1-FN)> 0
\end{equation}
}}
Similarly rewriting$(f(\omega_1)-f(\omega_2)) > 0$, we get 
{\scalefont{0.9}{  
\begin{equation}\label{eq:case_one_2}
\Big(\Big[1-\prod_{v_i \in \omega_1(B)} (1-p_{\sD}(v_i))\Big]- \Big[1-\prod_{v_j \in \omega_2(B)} (1-p_{\sD}(v_j))\Big]\Big)FP > 0
\end{equation}
}}
As $0 < FP <1$ and $0 < FN < 1$, Eqs.~\eqref{eq:case_one_1} and~\eqref{eq:case_one_2} imply 
\begin{eqnarray}
\Big[\prod_{v_i \in \omega_2(A)} (1-p_{\sD}(v_i)) -\prod_{v_j \in \omega_1(A)} (1-p_{\sD}(v_i))\Big]> 0\label{eq:case_one_3}\\
\Big[\prod_{v_i \in \omega_2(B)} (1-p_{\sD}(v_i)) -\prod_{v_j \in \omega_1(B)} (1-p_{\sD}(v_j))\Big]> 0\label{eq:case_one_4}
\end{eqnarray}
Eqs.~\eqref{eq:case_one_3} and~\eqref{eq:case_one_4} imply 
\begin{eqnarray}
\prod_{v_i \in \omega_2(A)} (1-p_{\sD}(v_i)) &>& \prod_{v_j \in \omega_1(A)} (1-p_{\sD}(v_i)) \mbox{~and}\label{eq:case_one_5}\\
\prod_{v_i \in \omega_2(B)} (1-p_{\sD}(v_i)) &>& \prod_{v_j \in \omega_1(B)} (1-p_{\sD}(v_j))\label{eq:case_one_6}
\end{eqnarray}
Note that at every state in ${\bf s} \in \bS$ with ${\bf s}=\{v_{i_1}, \ldots, v_{i_k} \}$ $0 \leqslant \sum_{v_{i_k} \in {\bf s}}p_{\sD}(v_{i_k}) \leqslant 1$. Thus Eqs.~\eqref{eq:case_one_5} and~\eqref{eq:case_one_6} cannot be together satisfied. Thus case~(i) does not hold at an NE. Following the similar arguments one can show that  case~(ii) also does not hold at an NE. This implies  $(p(\omega_1)-p(\omega_2)) = 0$ and $(f(\omega_1)-f(\omega_2))= 0$.

Since $\omega_1$ and $\omega_2$ are arbitrary, one can show that for the general case 
\begin{equation}\label{eq:GlobalUa}
\sum_{i=1}^{|\Omega|}\epsilon_ip(\omega_i) = 0 \mbox{~and~} \sum_{i=1}^{|\Omega|}\epsilon_i f(\omega_i) = 0.
\end{equation}

Eq.~\eqref{eq:GlobalUa} should hold for all possible values of $\epsilon_i$'s satisfying 
$\sum_{i=1}^{|\Omega|}\epsilon_i = 0.$ This gives  $p(\omega_i) = p(\omega_j)$ and $f(\omega_i) = f(\omega_j)$, for all $i, j \in \{1, \ldots , |\Omega|\}$, at NE. This completes the proof.
\qed

Theorem~\ref{thm:mincutNE} concludes that the set of all solutions to the min-cut problem characterizes the set of NE of game ${\bf G}$. Two key properties of an equilibrium strategy pair is that defender chooses the min-cut nodes as trapping nodes and indeed with equal probability, which is proved in Lemma~\ref{lem:P_d_equal}. 


\begin{lemma}\label{lem:P_d_equal}
Consider the APT vs. DIFT game ${\bf G}$ and let $(p_{\sA}, p_{\sD})$  be an NE of ${\bf G}$. Then under $p_{\sD}$ all the min-cut nodes of the flow-network have equal probability of selecting as a trapping node.
\end{lemma}
\begin{proof}
At NE, all the paths in the state space have equal probability of getting detected, i.e., $p(\omega)$'s are same for all $\omega \in \Omega$ (by Theorem~\ref{thm:mincutNE}).  For a path $\omega \in \Omega$ with $\omega(A)$ denoting the set of nodes corresponding to the adversarial flow, we know
\begin{equation}\label{eq:p}
p(\omega) = \Big[1-\prod_{v_i \in \omega(A)} (1-p_{\sD}(v_i))\Big](1-FN), \mbox{~for~all~} \omega \in \Omega.
\end{equation}
Also, there is exactly one node corresponding to an attack path  that has nonzero probability of selecting as a trapping node (Theorem~\ref{thm:mincutNE}). Hence in set $\omega(A)$ exactly one node, say $v_i \in \omega(A)$,  has nonzero value of $p_{\sD}$. Eq.~\eqref{eq:p} hence implies that at NE all the min-cut nodes of the flow-network have equal probability of selecting as a trapping node.
\end{proof}

As a consequence of Theorem~\ref{thm:mincutNE} and Lemma~\ref{lem:P_d_equal}, the following corollary holds.
\begin{cor}\label{cor:equal}
Let $ \hat{\S}^\*$ be a set of min-cut nodes of the flow-network $\overbar{\F}$. Then, at an NE of the game ${\bf G}$, the  defender's strategy is to choose all the nodes in $ \hat{\S}^\*$ as trapping nodes with \underline{equal probability} and the adversary's strategy is to choose transitions in such a way that each attack path passes through exactly one node in $ \hat{\S}^\*$. 
\end{cor}
\begin{proof}
By Theorem~\ref{thm:mincutNE} and Lemma~\ref{lem:P_d_equal} the defender's strategy at NE is to select min-cut nodes as trapping nodes with equal probability. We know adversary's payoff
\begin{eqnarray*}
U_{\sA}(p_{\sD}, p_{\sA}) &=& \sum_{\omega \in \Omega}\pi(\omega)\Big[p(\omega)\,\alpha_{\sA} + (1-p(\omega) +f(\omega))\beta_{\sA} \Big],
\end{eqnarray*}
and at NE $p(\omega)$'s and $f(\omega)$'s are equal for all $\omega \in \Omega$. Thus the adversary's optimal strategy is to select any path (or set of paths if mixed)  that passes through exactly one min-cut node. 
\end{proof}
Using Theorem~\ref{thm:mincutNE} and Corollary~\ref{cor:equal}  we present below an approach to compute an equilibrium strategy pair of ${\bf G}$.
 Firstly, we solve the node version of the (source-sink)-min-cut problem on $\F$. Let an optimal solution be $\kappa(\hat{\S}^\*)$ and the corresponding vertex set  be $\hat{\S^\*} = \{\tilde{v}_{1}, \ldots, \tilde{v}_r  \}$, where $\hat{\S^\*} := \{v_i: (v_i, v'_i) \in \kappa(\hat{\S}^\*)\}$.  
By Theorem~\ref{thm:mincutNE}, at NE the defender only selects the nodes in $\hat{\S^\*}$ as trapping nodes, with equal probability, and the adversary chooses transitions such that each attack path passes through only one  node in $\hat{\S^\*}$. Therefore, the  attack paths  chosen by the adversary is characterized by the nodes in $\hat{\S^\*}$. The set of paths in the state space ${\bf S}$ can also be  characterized by the nodes in $\hat{\S^\*}$. In other words, the set of paths in the set $\Omega$ can be grouped such that each group corresponds to a set of paths in $\Omega$ in which attack path (not necessarily benign flows) passes through exactly one node in $\hat{\S^\*}$, i.e., at least one among the $W$ tagged flows passes through exactly one min-cut node.   We denote the set of paths in $\Omega$ that correspond to node $\tilde{v}_{i} \in \hat{\S^\*}$ as  $\Omega(\tilde{v}_{i})$. By Corollary~\ref{cor:equal} any distribution over paths in $\Omega(\tilde{v}_{i})$,  for $i\in \{1,\ldots, r\}$ is an optimal strategy for $\P_{\sA}$.

\begin{rem}\label{rem:path}
At NE, all attack paths in $\F$ with nonzero probability  pass through  exactly one node in the min-cut node set $\hat{\S^\*} = \{\tilde{v}_{1}, \ldots, \tilde{v}_r\}$. Further, the probabilities of selecting them as trapping nodes are also equal (since $p(\omega)$'s are equal).  Thus, without loss of generality, one can characterize the action space of the adversary as the set of  attack paths in $\F$ that pass through $\hat{\S^\*}$, disjoint with respect to nodes in $\hat{\S^\*}$, and the adversary strategizes over this set of  paths. 
\end{rem}

In the theorem below, we derive a closed-form expression for the optimal defender strategy for the game ${\bf G}$.
\begin{theorem}\label{th:NEmatrix}
Consider the APT vs. DIFT game ${\bf G}$. At NE the defender selects all the min-cut nodes as trapping nodes with probability $\dfrac{1}{\min\{ W, r\}}$, where $W, r$ respectively denote the number of tagged flows at a time in the system and the cardinality of the min-cut nodes.
\end{theorem}
\begin{proof}
Consider a non-terminal state ${\bf s} \in {\bf S}$, where ${\bf s}=\{v_{i_1}, \ldots, v_{i_W} \}$. The action set of $\P_{\sD}$ at ${\bf s}$ is $\A_{\sD}({\bf s})= \{0\}\cup \{v_{i_1}, \ldots, v_{i_W} \}$, for$ \{v_{i_1}, \ldots, v_{i_W} \} \subset V_{\G}$. Any defender strategy must satisfy for all ${\bf s} \in {\bf S}$, $p_{\sD}(0)+\sum_{j=1}^Wp_{\sD}(v_{i_j}) =1$ if $v_{i_j} \neq \phi$. By Corollary~\ref{cor:equal}, at NE, $\P_{\sD}$ only selects min-cut nodes as trapping nodes. Let $\hat{\S^\*}= \{\tilde{v}_{1}, \ldots, \tilde{v}_r  \}$ be the solution obtained for the min-cut problem. Thus the constraint on $p_{\sD}$ (that it should add upto 1 at all states) boils down to the following.  For all states ${\bf s}=\{v_{i_1}, \ldots, v_{i_W} \} \subseteq  \{\tilde{v}_{1}, \ldots, \tilde{v}_r\}$,  
\begin{equation}\label{eq:add_1}
p_{\sD}(0)+\sum_{\substack{v_{i_j}: v_{i_j} \in \{\tilde{v}_{1}, \ldots, \tilde{v}_r\}\\j \in \{1, \ldots, W\}}}p_{\sD}(v_{i_j}) =1, 
\end{equation}
 Moreover, we know all min-cut nodes are chosen as trapping nodes with equal probability (Corollary~\ref{cor:equal}), i.e., $p_{\sD}(\tilde{v}_{1})=\ldots=p_{\sD}(\tilde{v}_{r}) := \theta$.  For all states ${\bf s}=\{v_{i_1}, \ldots, v_{i_W} \} \subseteq  \{\tilde{v}_{1}, \ldots, \tilde{v}_r\}$
\begin{equation}\label{eq:add_2}
p_{\sD}(0)+\theta\Big(\sum_{\substack{v_{i_j}: v_{i_j} \in \{\tilde{v}_{1}, \ldots, \tilde{v}_r\}\\j \in \{1, \ldots, W\}}}1\Big) =1, 
\end{equation}
For all states ${\bf s} \in {\bf S}$ with ${\bf s}=\{v_{i_1}, \ldots, v_{i_W} \} \subseteq  \{\tilde{v}_{1}, \ldots, \tilde{v}_r\}$, 
\begin{equation}\label{eq:add_3}
\sum_{\substack{v_{i_j}: v_{i_j} \in \{\tilde{v}_{1}, \ldots, \tilde{v}_r\}\\j \in \{1, \ldots, W\}}}1 \leqslant \min\{W, r\}.
\end{equation}
Eq.~\eqref{eq:add_3} must hold for all states and hence the maximum value that $\theta$ can take is  when   $p_{\sD}(0)=0$ and  $\sum_{\substack{v_{i_j}: v_{i_j} \in \{\tilde{v}_{1}, \ldots, \tilde{v}_r\}\\j \in \{1, \ldots, W\}}}1 = \min\{W, r\}$ at a state ${\bf s} \in {\bf S}$. Thus
$\theta = \dfrac{1}{\min\{W, r\}}.$

As the defender's and adversary's action sets are characterized by the min-cut nodes, the defender's payoff is
{\scalefont{0.95}{
\begin{eqnarray}
U_{\sD}(p_{\sD}, p_{\sA}) &=&\hspace*{-3 mm} \sum_{i=1}^r\sum_{\omega \in \Omega(\tilde{v}_{i})}\hspace*{-2 mm}\pi(\omega)\Big[p(\omega)\,\alpha_{\sD} + (1-p(\omega) +f(\omega))\beta_{\sD} \Big]\nonumber\\
&+& \sum_{i=1}^r\sum_{\omega \in \Omega(\tilde{v}_{i})}\Big(  \sum_{\tilde{v}_i \in \omega}p_{\sD}(\tilde{v}_i)\C_{\sD}(\tilde{v}_i)\Big)\nonumber\\
&=&\hspace*{-3 mm} \sum_{i=1}^r\sum_{\omega \in \Omega(\tilde{v}_{i})}\hspace*{-2 mm}\pi(\omega)\Big[\theta(1-FN)\,(\alpha_{\sD}-\beta_{\sD}) + \beta_{\sD}+ f(\omega)\beta_{\sD} \Big]\nonumber\\
&+& \sum_{i=1}^r\sum_{\omega \in \Omega(\tilde{v}_{i})}\Big(  \sum_{\tilde{v}_i \in \omega}\theta\C_{\sD}(\tilde{v}_i)\Big)\label{eq:def_theta}
\end{eqnarray}
}}
Eq.~\eqref{eq:def_theta} holds from Lemma~\ref{lem:P_d_equal} by substituting $p(\omega)= [1-\prod_{\tilde{v}_i \in \omega(A)} (1-p_{\sD}(\tilde{v}_i))](1-FN)=(1-(1-\theta))(1-FN)=\theta(1-FN)$.   Here $0 \leqslant \theta(1-FN) \leqslant 1$ and $0 \leqslant f(\omega) \leqslant 1$. Also,  $(\alpha_{\sD}-\beta_{\sD})  >> \beta_{\sD}$ with $(\alpha_{\sD}-\beta_{\sD})  >>1$ and $\beta_{\sD} < 0$. Also note that any distribution, $\pi(\omega)$, over paths in $\Omega(\tilde{v}_{i})$, for $i\in\{1,\ldots, r\}$, is optimal for adversary (Corollary~\ref{cor:equal}).  Thus Eq.~\eqref{eq:def_theta} is maximum for the maximum possible $\theta$. Hence at NE the defender selects a set of min-cut nodes as trapping nodes with probability $1/\min\{W, r\}$.
\end{proof}
In this work, we provide a characterization for the equilibrium of the DIFT vs. APT game using the min-cut analysis. Min-cut of a flow-network need not necessarily be unique. For an application whose IFG results in a flow-network with multiple min-cut sets, equilibrium of the DIFT vs. APT game will be nonunique.  The game considers a strategic adversary and the defense strategy computed as equilibrium solution to the game provides a selection of the trapping nodes that maximizes the probability of detection while minimizing the cost of security analysis. During a real attack scenario, the adversary may choose an attack path which is different from the equilibrium of the game. We note that, by  definition of Nash equilibrium, such a strategy will not improve the payoff of the adversary.  Moreover, for a defender that is unaware of the actions of the adversary,  implementing a solution  to the game as the trapping nodes  maximizes the probability of detection.
\section{Numerical Study}\label{sec:sim}
We validate our theoretical results using real-world attack data obtained using RAIN \cite{JiLeeDowWanFazKimOrsLee-17} for a data exfiltration attack. We use  system log data collected during DARPA Transparent computing (DARPA-TC) red team engagement that evaluated the RAIN log recording system. A brief description of the dataset  used and the steps involved in the construction of the IFG for that attack is given below.

The goal of the attack is to exfiltrate sensitive information from the system.
The attack uses stolen credentials to exfiltrate sensitive files from the
targeted machine. The attack exfiltrates sensitive system files using \textit{scp}. The resulting IFG  of the data consists of $299$ nodes and $6398$ edges. Figure~\ref{fig:3} shows a portion of the IFG, obtained after performing a pruning technique, that captures all possible attack paths in the IFG. The details of the type of the nodes in the IFG in Figure~\ref{fig:3} is presented in Table~\ref{tab:IFG}. Entry points of the attack are identified as \textit{/bin/bash} and \textit{/usr/sbin/sshd} (nodes~1 and~5 in Figure~\ref{fig:3}). Destination node of the attack is \textit{/etc/passwd} (node~3 in Figure~\ref{fig:3}).

\begin{figure}[h]
	\centering
	{\includegraphics[width=0.4\textwidth]{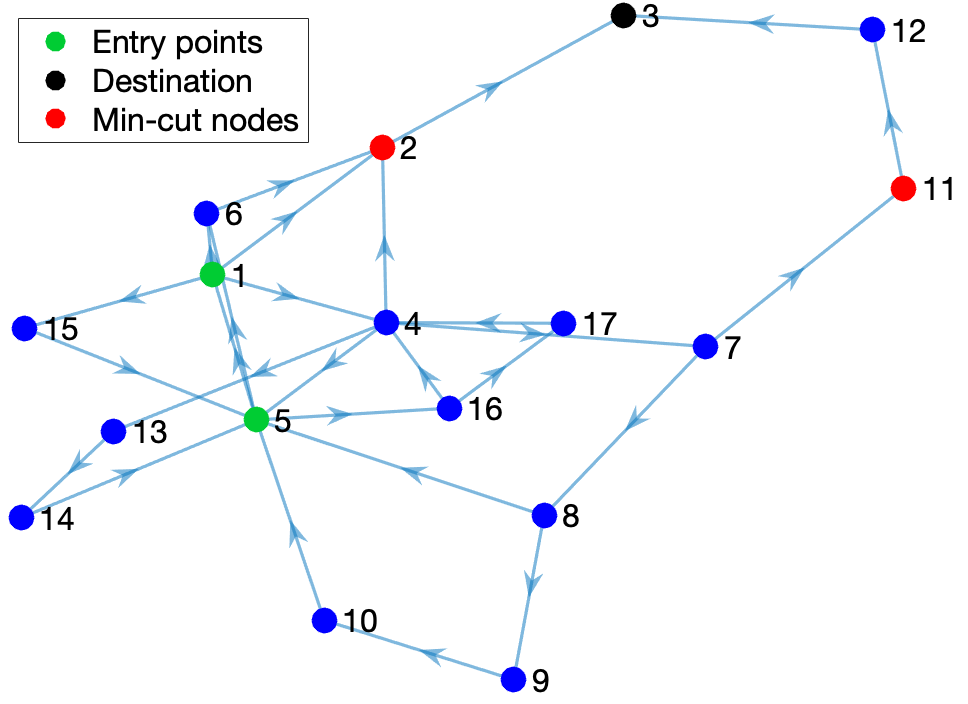}}
	\hspace{\parindent}
	\vspace*{-3 mm}
	\caption{\small Portion of the IFG that captures all  attack paths in IFG.}\label{fig:3}
\end{figure} 

\begin{table}[t]
\begin{minipage}{9cm}
 \centering
\captionof{table}{Description of the nodes of the IFG for the  data exfiltration attack} 
\resizebox{\textwidth}{!}{%
\begin{tabular}{|c|c|c|}
\hline
\multicolumn{1}{|l|} {\bf Node ID} & \multicolumn{1}{|l|}  {\bf Node Name} & \multicolumn{1}{|l|}{\bf Node Type}\\
\hline
\multicolumn{1}{|l|} {1} & \multicolumn{1}{|l|}  {/bin/bash} & \multicolumn{1}{|l|}{Process}\\
\hline
\multicolumn{1}{|l|} {2} & \multicolumn{1}{|l|}  {/usr/bin/sudo} & \multicolumn{1}{|l|}{Process}\\
\hline
\multicolumn{1}{|l|} {3} & \multicolumn{1}{|l|} {/etc/passwd} & \multicolumn{1}{|l|}{Files}\\
\hline	
\multicolumn{1}{|l|} {4} & \multicolumn{1}{|l|} {Unknown} & \multicolumn{1}{|l|}{IPC Object} \\
\hline
\multicolumn{1}{|l|} {5} &\multicolumn{1}{|l|}  {/usr/sbin/sshd} & \multicolumn{1}{|l|}{Process}\\
\hline
\multicolumn{1}{|l|} {6} & \multicolumn{1}{|l|} { /etc/group} & \multicolumn{1}{|l|}{File}\\
\hline
\multicolumn{1}{|l|} {7} & \multicolumn{1}{|l|} { /usr/sbin/console-kit-daemon} & \multicolumn{1}{|l|}{Process}\\
\hline
\multicolumn{1}{|l|} {8} & \multicolumn{1}{|l|}  {Unknown} & \multicolumn{1}{|l|}{IPC Object}\\
\hline
\multicolumn{1}{|l|} {9} & \multicolumn{1}{|l|}  {/usr/sbin/avahi-daemon} & \multicolumn{1}{|l|}{Process}\\
\hline
\multicolumn{1}{|l|} {10} & \multicolumn{1}{|l|} {Unknown} & \multicolumn{1}{|l|}{IPC Object}\\
\hline
\multicolumn{1}{|l|} {11} & \multicolumn{1}{|l|} {/run/ConsoleKit/database$^{\thicksim}$} & \multicolumn{1}{|l|}{File}\\
\hline
\multicolumn{1}{|l|} {12} & \multicolumn{1}{|l|} {/usr/lib/policykit-1/polkitd} & \multicolumn{1}{|l|}{Process}\\
\hline
\multicolumn{1}{|l|} {13} & \multicolumn{1}{|l|} {/bin/run-parts} & \multicolumn{1}{|l|}{File}\\
\hline
\multicolumn{1}{|l|} {14} & \multicolumn{1}{|l|} {/run/motd.new} & \multicolumn{1}{|l|}{File}\\
\hline
\multicolumn{1}{|l|} {15} & \multicolumn{1}{|l|} {/home/theia/secrets.tar.gz} & \multicolumn{1}{|l|}{File}\\
\hline
\multicolumn{1}{|l|} {16} &\multicolumn{1}{|l|}  {/bin/dash} & \multicolumn{1}{|l|}{File}\\
\hline
\multicolumn{1}{|l|} {17} & \multicolumn{1}{|l|} {/usr/bin/apt-config} & \multicolumn{1}{|l|}{File}\\
\hline
\end{tabular} }
\label{tab:IFG}
\end{minipage}
\end{table}

First we solve the min-cut problem on the flow-network constructed from the IFG of the data. Resulting min-cut nodes, \textit{/usr/bin/sudo} and \textit{/run/consolekit/data}, are indicated in red color in Figure~\ref{fig:3}. The min-cut nodes are chosen as trapping nodes with probability $\theta = 1/\min\{W,r\}$. $W$ and $r$ denote the number of tagged flows at a time in the system and the cardinality of the obtained solution to the min-cut problem, respectively. Here $W = 3$ and $r = 2$. Thus $\theta = 0.5$.
  
\begin{figure*}[t]
\centering
\begin{subfigure}[b]{0.41\textwidth}
\begin{tikzpicture}[scale=0.8]
\pgfplotsset{compat=1.11,
    /pgfplots/ybar legend/.style={
    /pgfplots/legend image code/.code={%
       \draw[##1,/tikz/.cd,yshift=-0.25em]
        (0cm,0cm) rectangle (3pt,0.8em);},
   },
}
\begin{axis}[
    ybar,
     x tick label style  = {color=black,text width=1.5cm,align=center, font=\small},
    enlargelimits=0.15,            
    legend style={at={(0.5,1)},
                anchor=north,legend columns=-1},    
    ylabel={Defender's payoff},
    symbolic x coords={1,2,3, 4},
      xticklabels   = {$FP=0$ $FN=0$, $FP=0$ $FN=0.2$, $FP=0.2$ $FN=0.2$, $FP=0.2$ $FN=0$},
    xtick=data,
    nodes near coords,
    every node near coord/.append style={font=\tiny},
    nodes near coords align={vertical},
    ]
\addplot coordinates {(1,-19.1) (2,-232.9) (3,-280.4) (4,-100.7)};
\addplot coordinates {(1,-223.7) (2,-421.7) (3,-439.1)(4,-295)};
\addplot coordinates {(1,-636.7) (2,-718) (3,-676.8)(4,-618.6)};
\legend{{$\theta=0.5$},{$\theta=0.4$},{$\theta=0.2$}}
\end{axis}
\end{tikzpicture}
    \caption{}\label{fig:bar1}
\end{subfigure}~\hspace{1 cm}
\begin{subfigure}[b]{0.41\textwidth}
\centering
\begin{tikzpicture}[scale=0.8]
\pgfplotsset{compat=1.11,
    /pgfplots/ybar legend/.style={
    /pgfplots/legend image code/.code={%
       \draw[##1,/tikz/.cd,yshift=-0.25em]
        (0cm,0cm) rectangle (3pt,0.8em);},
   },
}
\begin{axis}[
    ybar,
     x tick label style  = {color=black,text width=1.5cm,align=center, font=\small},
    enlargelimits=0.15,            
    legend style={at={(0.5,1)},
                anchor=north,legend columns=-1},    
    ylabel={Adversary's payoff},
    symbolic x coords={1,2,3, 4},
      xticklabels   = {$FP=0$ $FN=0$, $FP=0$ $FN=0.2$, $FP=0.2$ $FN=0.2$, $FP=0.2$ $FN=0$},
    xtick=data,
    nodes near coords,
    every node near coord/.append style={font=\tiny},
    nodes near coords align={vertical},
    ]
\addplot coordinates {(1,12) (2,226) (3,274) (4,94)};
\addplot coordinates {(1,218) (2,416) (3,434)(4,290)};
\addplot coordinates {(1,634) (2,716) (3,674)(4,616)};
\legend{{$\theta=0.5$},{$\theta=0.4$},{$\theta=0.2$}}
\end{axis}
\end{tikzpicture}
\caption{}\label{fig:bar2}
\end{subfigure}~\hspace{0.2 cm}
\caption{\small The parameters chosen are: $\alpha_{\sA}=-1000$, $\beta_{\sA}=1000$, $\alpha_{\sD}=1000$, $\beta_{\sA}=-1000$. The resource cost for the nodes are chosen such that $\C_{\sD}(v_i)$ is proportional to the total number of information flows at node $v_i$ for the whole logging period, where $i\in \{1, \ldots, N\}$. Figures~\ref{fig:4}~(a) and~\ref{fig:4}~(b) shows the payoffs of the defender and adversary, respectively, for different values of false-positives ($FP$) and false-negatives ($FN$) when min-cut nodes are chosen as trapping nodes. The probability of selecting a min-cut nodes as trapping node  is varied from $\theta=1/\min\{W,r\}=1/\min\{3,2\}=0.5$, $\theta=0.4$ and $\theta=0.2$ for all the experiments. Also, each case is averaged over $1000$ runs.}\label{fig:4}
\end{figure*}
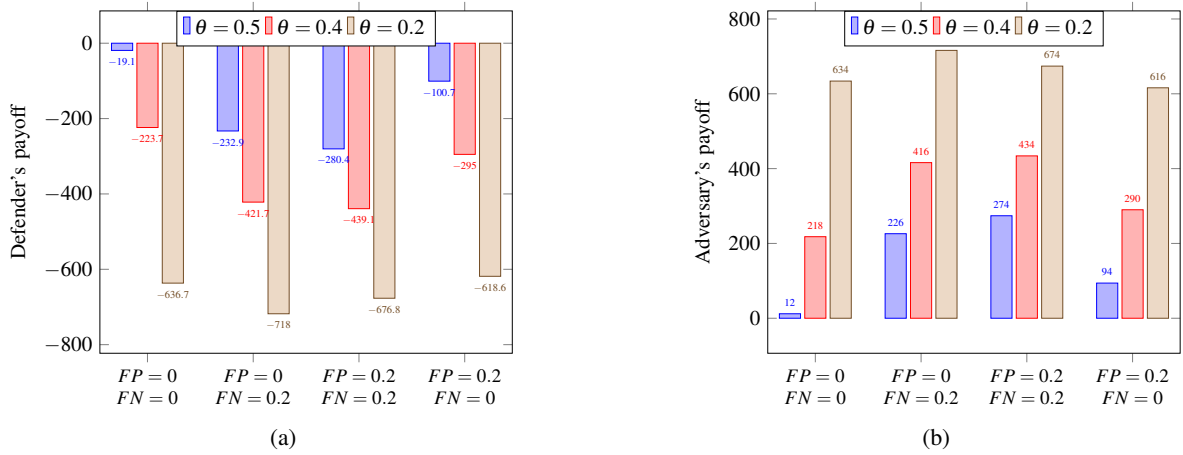
\begin{figure*}[t]
\centering
\begin{subfigure}[b]{0.41\textwidth}
\begin{tikzpicture}[scale=0.8]
\pgfplotsset{compat=1.11,
    /pgfplots/ybar legend/.style={
    /pgfplots/legend image code/.code={%
       \draw[##1,/tikz/.cd,yshift=-0.25em]
        (0cm,0cm) rectangle (3pt,0.8em);},
   },
}
\begin{axis}[
    ybar,
     x tick label style  = {color=black,text width=1.5cm,align=center, font=\small},
    enlargelimits=0.15,            
    legend style={at={(0.5,1)},
                anchor=north,legend columns=-1},    
    ylabel={Defender's payoff},
    symbolic x coords={1,2,3, 4},
      xticklabels   = {$FP=0$ $FN=0$, $FP=0$ $FN=0.2$, $FP=0.2$ $FN=0.2$, $FP=0.2$ $FN=0$},
    xtick=data,
    nodes near coords,
    every node near coord/.append style={font=\tiny},
    nodes near coords align={vertical},
    ]
\addplot coordinates {(1,-19.1) (2,-232.9) (3,-280.4) (4,-100.7)};
\addplot coordinates {(1,-446.4) (2,-600.8) (3,-984.2)(4,-914.6)};
\addplot coordinates {(1,-1751.1) (2,-1873.1) (3,-1594.5)(4,-1545.4)};
\legend{{Min-cut},{Cut},{Not a cut}}
\end{axis}
\end{tikzpicture}
    \caption{}\label{fig:bar3}
\end{subfigure}~\hspace{1 cm}
\begin{subfigure}[b]{0.41\textwidth}
\centering
\begin{tikzpicture}[scale=0.8]
\pgfplotsset{compat=1.11,
    /pgfplots/ybar legend/.style={
    /pgfplots/legend image code/.code={%
       \draw[##1,/tikz/.cd,yshift=-0.25em]
        (0cm,0cm) rectangle (3pt,0.8em);},
   },
}
\begin{axis}[
    ybar,
     x tick label style  = {color=black,text width=1.5cm,align=center, font=\small},
    enlargelimits=0.15,            
    legend style={at={(0.5,1)},
                anchor=north,legend columns=-1},    
    ylabel={Adversary's payoff},
    symbolic x coords={1,2,3, 4},
      xticklabels   = {$FP=0$ $FN=0$, $FP=0$ $FN=0.2$, $FP=0.2$ $FN=0.2$, $FP=0.2$ $FN=0$},
    xtick=data,
    nodes near coords,
    every node near coord/.append style={font=\tiny},
    nodes near coords align={vertical},
    ]
\addplot coordinates {(1,12) (2,226) (3,274) (4,94)};
\addplot coordinates {(1,84) (2,234) (3,768)(4,698)};
\addplot coordinates {(1,306) (2,424) (3,928)(4,876)};
\legend{{Min-cut},{Cut},{Not a cut}}
\end{axis}
\end{tikzpicture}
\caption{}\label{fig:bar4}
\end{subfigure}~\hspace{0.2 cm}
\caption{\small The parameters chosen are: $\alpha_{\sA}=-1000$, $\beta_{\sA}=1000$, $\alpha_{\sD}=1000$, $\beta_{\sA}=-1000$. The resource cost for the nodes are chosen such that $\C_{\sD}(v_i)$ is proportional to the total number of information flows at node $v_i$ for the whole logging period, where $i\in \{1, \ldots, N\}$. Figures~\ref{fig:5}~(a) and~\ref{fig:5}~(b) shows the payoffs of the defender and adversary, respectively, for different values of false-positives ($FP$) and false-negatives ($FN$) when (i)~min-cut nodes are chosen as trapping nodes, (ii)~cut nodes (not min-cut), and~(iii)~nodes that are not a cut.  Also, each case is averaged over $1000$ runs.}\label{fig:5}
\end{figure*}

Then we simulate an attack in the IFG and perform security analysis to evaluate the performance of the DIFT model. We conduct two experiments which are detailed below. Each payoff value represented in Figure~\ref{fig:4} and Figure~\ref{fig:5} are obtained after averaging over $1000$ trials. We have identified $13$ distinct attack paths related to the adversary in the underlying IFG and uniformly picked one attack path in each trial. Each of these attack paths consist only one min-cut node (Theorem~\ref{thm:mincutNE}). Further, we set $\alpha_{\sA}=-1000$, $\beta_{\sA}=1000$, $\alpha_{\sD}=1000$, and $\beta_{\sA}=-1000$. $\C_{\sD}(v_i)$ values are set to be proportional to the number of information flows passing through each node, $v_i$ for $i \in \{1, \ldots, N\}$, in the IFG through out the logging period. 
\subsection{Case Study 1}\label{study1}
In this experiment setup, we vary the probability of selecting min-cut nodes as trapping nodes. Recall that any probability greater than $\theta$ is not a valid defender's strategy (Theorem~\ref{th:NEmatrix}). Hence for comparing the performance we select probabilities $\theta \leqslant 1/\min\{W, r\} = 0.5$. Note that at NE $\theta = 0.5$. Figures~\ref{fig:4}~(a) and~\ref{fig:4}~(b) plot the payoff values for both players averaged over $1000$ trials for different values of false-positives and false-negatives. Figure~\ref{fig:4}~(a) shows that defender performs better when following NE strategy, i.e., $\theta =0.5$, which validates the theoretical analysis. 

\subsection{Case Study 2}
In this case study we compare the performance of the DIFT by varying the locations of trapping nodes. In provenance-based analysis, various heuristics have been proposed for selecting the trapping nodes  \cite{holmes}, \cite{heuristic-conf}. In \cite{holmes}, a heuristic that assigns weights to the nodes of the IFG is proposed. In this case study, we compare the performance of the defense strategy obtained as solution to the DIFT vs. APT game with two cases. For the first case, we assign nonzero weights to the cut nodes of the IFG, since cut nodes are the smallest set of nodes through which all the flows pass through, and assign zero weights to all non-cut nodes. For the second case, we select the non-cut nodes as the trapping nodes, since all information flows will pass through the non-cut nodes, by assigning nonzero weights to all non-cut nodes and assign zero weights to all cut nodes. 
Figures~\ref{fig:5}~(a) and~\ref{fig:5}~(b) show that the expected payoff of the defender takes higher values when the min-cut nodes are chosen as the trapping nodes. Also note that choosing a cut is better than selecting set of nodes that are not a cut.

\subsection{Case Study 3}
In this case study we compare the performance of the game theoretic approach with a heuristic rule-based method. For comparing the two approaches, we use the knowledge of the actual attack path as the ground truth. The  attack path for the data exfiltration attack is   $1\rightarrow 15\rightarrow 5 \rightarrow 16 \rightarrow 17 \rightarrow 4 \rightarrow 2 \rightarrow 3$ in the IFG given in Figure~\ref{fig:3}.  The heuristic rules used for identifying the trapping nodes are presented in Table~\ref{tab:rules}.   Using these rules, for the data  exfiltration attack, we identified the set of nodes $\{2, 6, 12, 13, 16,  17\}$ as the trapping nodes. Solution to the DIFT vs. APT game returns the min-cut node set $\{2, 11\}$  as the trapping nodes. In this case study, we set $W=1$ and consequently $\theta=1$. The parameters of the game are set as $\alpha_{\sD}=1$, $\beta_{\sD}=-1$, $FN=0.2$, and $FP=0$. 

Resource cost for performing security analysis at a node $v_i$ is given by
\begin{equation}
\C_{\sD}(v_i) = c~q(v_i),
\end{equation}
where $c$ is the cost parameter for performing security analysis and $q(v_i)$ is the fraction of information flows passing through  node $v_i$. The resource cost for conducting security analysis is more for a node with higher fraction of flows, i.e., a busy node. Figure~\ref{fig:heuristic} presents a plot showing the variation of the DIFT's payoff with respect to $c$.  For the data exfiltration attack considered and for the selection of the trapping nodes $\{2, 6, 12, 13, 16,  17\}$, heuristic rule-based method will conduct security analysis on the flow at nodes $2, 16$, and $17$. On the other hand, the game theoretic modeling will conduct security analysis on the flow  at node $2$. Table~\ref{prob} presents the fraction of the information flows passing through nodes $2, 16,$ and $17$. Figure~\ref{fig:heuristic} shows that the hueristic rule-based method performs better for lower values of $c$. This is expected since when cost for conducting security analysis is negligible it is better to analyze the flow at multiple locations, as it increases the probability of detection. As the cost parameter increases,  the game theoretic solution outperforms the heuristic. The performance of the heuristic decreases drastically for higher values of $c$ as heuristic conducts security analysis at node $16$ which is a busier node and hence incurs high resource cost. The game, on the other hand, considers the resource cost while computing the min-cut and hence chooses node $2$ for conducting the security analysis. This validates that the game theoretic approach provides a trade-off between the effectiveness of detection and the cost of detection. 

\begin{table}[h!]
	\centering
		\caption{Fraction of flows traversing nodes in IFG}
	\begin{tabular}{|c|c|c|c|c|c|c|} 
		\hline
		Node $v_i$ & 2 & 16 & 17  \\
		\hline 
		$q(v_i)$ &  0.005705 & 0.167162 & 0.016177  \\ 
		\hline 
	\end{tabular}
	\label{prob}
\end{table}

\begin{table}[t]
\begin{minipage}{9cm}
 \centering
\captionof{table}{Heuristic rules used for identifying trapping nodes in Case study 3} 
\resizebox{\textwidth}{!}{%
\begin{tabular}{|c|c|c|}
\hline
\multicolumn{1}{|l|} {\bf System Call} & \multicolumn{1}{|l|}  {\bf Performed on File} & \multicolumn{1}{|l|}{\bf Description}\\
\hline
\multicolumn{1}{|l|} {read} & \multicolumn{1}{|l|}  {/dev/mem} &  \multirow{1}{*} {Potentially }\\\cline{1-2} 
\multicolumn{1}{|l|} { read} &  \multicolumn{1}{|l|}{/etc/passwd} & anomalous actions \\ \cline{1-2} 
 \multicolumn{1}{|l|} {read} & \multicolumn{1}{|l|}{/etc/shadow} & on sensitive files/folder \\
\hline	
\multirow{1}{*} {} &\multirow{1}{*} {} & \multirow{1}{*} {Potentially} \\ 
\multicolumn{1}{|l|} {read}& \multicolumn{1}{|l|} {*.ssh/id\_rsa}  & {anomalous actions}\\
& & {on sensitive SSH files} \\
\hline
\multicolumn{1}{|l|} {read} &\multicolumn{1}{|l|}  {*.default/key3.db} & \multirow{17}{*} {Potentially } \\ \cline{1-2} 
\multicolumn{1}{|l|} {read} & \multicolumn{1}{|l|} {*.default/logins.json} & \multirow{18}{*} {anomalous actions} \\ \cline{1-2} 
\multicolumn{1}{|l|} {read} & \multicolumn{1}{|l|} {*.default/signons.sqlite} & \multirow{19}{*}{on sensitive SSH files} \\  \cline{1-2} 
\multicolumn{1}{|l|} {write} & \multicolumn{1}{|l|}  {/dev/mem} & \\  \cline{1-2} 
\multicolumn{1}{|l|} {write} & \multicolumn{1}{|l|}  {/etc/*} & \\  \cline{1-2} 
\multicolumn{1}{|l|} {write} & \multicolumn{1}{|l|} {/usr/local/*} & \\  \cline{1-2} 
\multicolumn{1}{|l|} {write} & \multicolumn{1}{|l|} {/usr/local/bin/*} & \\  \cline{1-2} 
\multicolumn{1}{|l|} {write} & \multicolumn{1}{|l|} {/var/spool/cron/crontabs/root} & \\  \cline{1-2} 
\multicolumn{1}{|l|} {write} & \multicolumn{1}{|l|} {/bin/*} & \\  \cline{1-2} 
\multicolumn{1}{|l|} {write} & \multicolumn{1}{|l|} {/boot/*} &\\  \cline{1-2} 
\multicolumn{1}{|l|} {write} & \multicolumn{1}{|l|} {/var/*} &\\  \cline{1-2} 
\multicolumn{1}{|l|} {write} &\multicolumn{1}{|l|}  {/usr/bin/*} & \\  \cline{1-2} 
\multicolumn{1}{|l|} {write} & \multicolumn{1}{|l|} {/dev/*} &\\  \cline{1-2} 
\multicolumn{1}{|l|} {write} & \multicolumn{1}{|l|} {/etc/security/*} &\\  \cline{1-2} 
\multicolumn{1}{|l|} {write} & \multicolumn{1}{|l|} {/usr/spool/*} &\\ \cline{1-2} 
\multicolumn{1}{|l|} {write} & \multicolumn{1}{|l|} {/usr/etc/*} & \\ \cline{1-2} 
\multicolumn{1}{|l|} {write} & \multicolumn{1}{|l|} {/usr/kvm/*} &\\ \cline{1-2} 
\multicolumn{1}{|l|} {write} & \multicolumn{1}{|l|} {/usr/*} &\\ \cline{1-2} 
\multicolumn{1}{|l|} {write} & \multicolumn{1}{|l|} {/usr/lib/*} &\\ \cline{1-2} 
\multicolumn{1}{|l|} {write} & \multicolumn{1}{|l|} {*.default/key3.db} &\\ \cline{1-2} 
\multicolumn{1}{|l|} {write} & \multicolumn{1}{|l|} {*.default/logins.json} &\\ \cline{1-2} 
\multicolumn{1}{|l|} {write} & \multicolumn{1}{|l|} {*.default/signons.sqlite} & \\ \cline{1-2} 
\hline

\multirow{1}{*}   &  \multirow{1}{*}  &  \multirow{1}{*} {Important}\\
\multicolumn{1}{|l|}  {write} & \multicolumn{1}{|l|} {/etc/hosts}  &  {DNS information}\\\cline{1-2}
\hline
\multirow{1}{*}   &  \multirow{1}{*}  &  \multirow{1}{*} {Stores extensions}\\
\multicolumn{1}{|l|}  {write} & \multicolumn{1}{|l|} {*.default/extensions}  &  {for Firefox}\\\cline{1-2}
\hline
\multirow{1}{*}   &  \multirow{1}{*}  &  \multirow{1}{*} {Used for }\\
\multicolumn{1}{|l|}  {open} & \multicolumn{1}{|l|} {*libpam*}  &  {authentication}\\\cline{1-2}
\hline

\multicolumn{1}{|l|} {exec} & \multicolumn{1}{|l|}  {bin/ip} &  \multirow{1}{*} {Used for gathering  }\\\cline{1-2} 
\multicolumn{1}{|l|} { exec} &  \multicolumn{1}{|l|}{/bin/ifconfig} & information about  \\ \cline{1-2} 
 \multicolumn{1}{|l|} {exec} & \multicolumn{1}{|l|}{/bin/netstat} & network \\
\hline	
\multirow{1}{*}   &  \multirow{1}{*}  &  \multirow{1}{*} {Changes permissions}\\
\multicolumn{1}{|l|}  {exec} & \multicolumn{1}{|l|} {/bin/chmod}  &  {of files/folders}\\\cline{1-2}
\hline
\multirow{1}{*}   &  \multirow{1}{*}  &  \multirow{1}{*} {Changes ownership}\\
\multicolumn{1}{|l|}  {exec} & \multicolumn{1}{|l|} {/bin/chown}  &  {of files/folders}\\\cline{1-2}
\hline 
\end{tabular} }
\label{tab:rules}
\end{minipage}
\end{table}
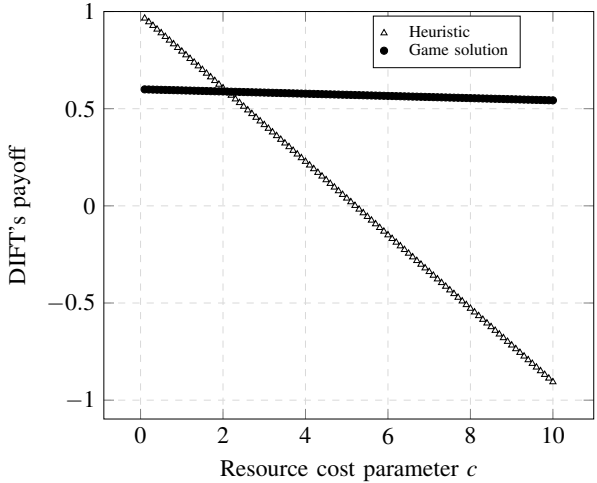
\begin{figure}[t]
\centering
    \begin{tikzpicture}[scale=0.9]
        \centering
      \begin{axis}[
          width=\linewidth, 
          grid=major, 
          grid style={dashed,gray!30}, 
          xlabel={ Resource cost parameter $c$}, 
          ylabel={ DIFT's payoff},
          legend style={at={(0.7,0.99)},anchor=north,legend cell align=left,row sep=-0.1cm,column sep=7pt,font=\scriptsize}, 
              smooth,
             ymax=1
        ]
        \addplot[only marks,color=black,mark=triangle,mark size=1.5pt]  table[x index = {0}, y index = {1},col sep=comma] {heuristic};
        \addlegendentry{Heuristic}
	\addplot[only marks,color=black,mark=*,mark size=1.5pt]  table[x index = {0}, y index = {2},col sep=comma] {heuristic};
        \addlegendentry{Game solution}      
      \end{axis}
    \end{tikzpicture}

    \caption{\small Plot shows variation of DIFT's payoff for  (i)~min-cut solution obtained as equilibrium of the  game and (ii)~heuristic rules-based method,  with respect to  $c$.}
    \label{fig:heuristic}
\end{figure}

\section{Conclusion}\label{sec:conclu}
This paper provided a game-theoretic model for the detection of Advanced Persistent Threats (APTs) using Dynamic Information Flow Tracking (DIFT). The interaction of  APT and  DIFT is modeled as a stochastic game in which the defender dynamically chooses locations (trapping nodes)  to perform security analysis  and the adversary chooses transitions along attack paths in order to maximize the probability of reaching the desired target. Before performing security analysis, the defender is unable to distinguish whether a tagged information flow at a node in the information flow graph is malicious or not and this results in the information asymmetry of the players. Further, the defender is not accurate in its detection process, which results in the generation of false-positives and false-negatives.  We  modeled the strategic interactions between DIFT and APT  as  a nonzero-sum, turn-based stochastic  game.  The game captures the information asymmetry among the players along with false-positives and false-negatives of DIFT.  Then we provided an approach to compute Nash equilibrium of the game, using a minimum capacity cut-set formulation.  Finally, we implemented our algorithm on real-world data for a data exfiltration attack obtained using the Refinable Attack INvestigation (RAIN) system.  Extending the game model to address a multi-attacker case with multiple malicious flows is a part of future work. 
\bibliographystyle{myIEEEtran}
\bibliography{MURI_References}

\begin{thebibliography}{10}
\providecommand{\url}[1]{#1}
\csname url@rmstyle\endcsname
\providecommand{\newblock}{\relax}
\providecommand{\bibinfo}[2]{#2}
\providecommand\BIBentrySTDinterwordspacing{\spaceskip=0pt\relax}
\providecommand\BIBentryALTinterwordstretchfactor{4}
\providecommand\BIBentryALTinterwordspacing{\spaceskip=\fontdimen2\font plus
\BIBentryALTinterwordstretchfactor\fontdimen3\font minus
  \fontdimen4\font\relax}
\providecommand\BIBforeignlanguage[2]{{%
\expandafter\ifx\csname l@#1\endcsname\relax
\typeout{** WARNING: IEEEtran.bst: No hyphenation pattern has been}%
\typeout{** loaded for the language `#1'. Using the pattern for}%
\typeout{** the default language instead.}%
\else
\language=\csname l@#1\endcsname
\fi
#2}}

\bibitem{JiLeeDowWanFazKimOrsLee-17}
Y.~Ji, S.~Lee, E.~Downing, W.~Wang, M.~Fazzini, T.~Kim, A.~Orso, and W.~Lee,
  ``{RAIN}: Refinable attack investigation with on-demand inter-process
  information flow tracking,'' \emph{ACM SIGSAC Conference on Computer and
  Communications Security}, pp. 377--390, 2017.

\bibitem{NewSon-05}
J.~Newsome and D.~Song, ``Dynamic taint analysis: Automatic detection,
  analysis, and signature generation of exploit attacks on commodity
  software,'' \emph{Network and Distributed Systems Security Symposium}, 2005.

\bibitem{DalKanKoz-08}
M.~Dalton, H.~Kannan, and C.~Kozyrakis, ``Real-world buffer overflow protection
  for userspace and kernelspace.'' \emph{USENIX Security Symposium}, pp.
  395--410, 2008.

\bibitem{DalKozZel-09}
M.~Dalton, C.~Kozyrakis, and N.~Zeldovich, ``Nemesis: Preventing authentication
  \& access control vulnerabilities in web applications,'' \emph{USENIX
  Security Symposium}, pp. 267--282, 2009.

\bibitem{Sha-53}
L.~S. Shapley, ``Stochastic games,'' \emph{Proceedings of the National Academy
  of Sciences}, vol.~39, no.~10, pp. 1095--1100, 1953.

\bibitem{LyeWin-05}
K.-w. Lye and J.~M. Wing, ``Game strategies in network security,''
  \emph{International Journal of Information Security}, vol.~4, no. 1-2, pp.
  71--86, 2005.

\bibitem{Ami-03}
R.~Amir, ``Stochastic games in economics and related fields: An overview,''
  \emph{Stochastic Games and Applications}, pp. 455--470, 2003.

\bibitem{ZhuBas-11}
Q.~Zhu and T.~Ba{\c{s}}ar, ``Robust and resilient control design for
  cyber-physical systems with an application to power systems,'' \emph{IEEE
  Decision and Control and European Control Conference (CDC-ECC)}, pp.
  4066--4071, 2011.

\bibitem{JasNow-16}
A.~Ja{\'s}kiewicz and A.~S. Nowak, ``Non-zero-sum stochastic games,''
  \emph{Handbook of Dynamic Game Theory}, pp. 1--64, 2016.

\bibitem{AlpBas-06}
T.~Alpcan and T.~Ba{\c{s}}ar, ``An intrusion detection game with limited
  observations,'' \emph{International Symposium on Dynamic Games and
  Applications}, vol.~26, 2006.

\bibitem{ZhuTemBas-10}
Q.~Zhu, H.~Tembine, and T.~Ba{\c{s}}ar, ``Network security configurations: A
  nonzero-sum stochastic game approach,'' \emph{American Control Conference
  (ACC)}, pp. 1059--1064, 2010.

\bibitem{NguAlpBas-09}
K.~C. Nguyen, T.~Alpcan, and T.~Ba{\c{s}}ar, ``Stochastic games for security in
  networks with interdependent nodes,'' \emph{International Conference on Game
  Theory for Networks}, pp. 697--703, 2009.

\bibitem{HesPra-01}
J.~P. Hespanha and M.~Prandini, ``{N}ash equilibria in partial-information
  games on {M}arkov chains,'' \emph{IEEE Conference on Decision and Control
  (CDC)}, vol.~3, pp. 2102--2107, 2001.

\bibitem{SahXiaClaLeePoo-18}
D.~Sahabandu, B.~Xiao, A.~Clark, S.~Lee, W.~Lee, and R.~Poovendran, ``{DIFT}
  games: dynamic information flow tracking games for advanced persistent
  threats,'' \emph{IEEE Conference on Decision and Control (CDC)}, pp.
  1136--1143, 2018.

\bibitem{MooSahClaLeePoo-18}
S.~Moothedath, D.~Sahabandu, A.~Clark, S.~Lee, W.~Lee, and R.~Poovendran,
  ``Multi-stage dynamic information flow tracking game,'' \emph{Conference on
  Decision and Game Theory for Security}, vol. 11199, pp. 80--101, 2018.

\bibitem{MooShaAllVClaBushWenPoo-18_arx}
S.~Moothedath, D.~Sahabandu, J.~Allen, A.~Clark, L.~Bushnell, W.~Lee, and
  R.~Poovendran, ``A game-theoretic approach for dynamic information flow
  tracking to detect multi-stage advanced persistent threats,'' \emph{IEEE
  Transactions on Automatic Control}, 2020.

\bibitem{dinukaCDC}
D.~Sahabandu, S.~Moothedath, J.~Allen, A.~Clark, L.~Bushnell, W.~Lee, and
  R.~Poovendran, ``Dynamic information flow tracking games for simultaneous
  detection of multiple attackers,'' \emph{IEEE Conference on Decision and
  Control (CDC)}, pp. 567--574, 2019.

\bibitem{Dinuka_ACC-19}
D.~Sahabandu, S.~Moothedath, J.~Allen, A.~Clark, L.~Bushnell, W.~Lee, and
  R.~Poovendran, ``A game theoretic approach for dynamic information flow
  tracking with conditional branching,'' \emph{American Control Conference
  (ACC)}, pp. 2289--2296, 2019.

\bibitem{stuxnet}
R.~Langner, ``Stuxnet: Dissecting a cyberwarfare weapon,'' \emph{IEEE Security
  \& Privacy}, vol.~9, no.~3, pp. 49--51, 2011.

\bibitem{SuhLeeZhaDev-04}
G.~E. Suh, J.~W. Lee, D.~Zhang, and S.~Devadas, ``Secure program execution via
  dynamic information flow tracking,'' \emph{ACM {SIGPLAN} Notices}, vol.~39,
  no.~11, pp. 85--96, 2004.

\bibitem{ClaLiOrs:07}
J.~Clause, W.~Li, and A.~Orso, ``Dytan: a generic dynamic taint analysis
  framework,'' \emph{International Symposium on Software Testing and Analysis},
  pp. 196--206, 2007.

\bibitem{hossain2018dependence}
M.~N. Hossain, J.~Wang, R.~Sekar, and S.~D. Stoller, ``Dependence-preserving
  data compaction for scalable forensic analysis,'' \emph{USENIX Security
  Symposium}, pp. 1723--1740, 2018.

\bibitem{log}
\BIBentryALTinterwordspacing
{Linux Kernel 3.7.1}. (2021, June 17). [Online]. Available:
  \url{https://docs.huihoo.com/doxygen/linux/kernel/3.7/timekeeping_8c.html#a43b6023a02e25bd465f61768418271b1}
\BIBentrySTDinterwordspacing

\bibitem{YinSonEgeChrEng-07}
H.~Yin, D.~Song, M.~Egele, C.~Kruegel, and E.~Kirda, ``Panorama: Capturing
  system-wide information flow for malware detection and analysis,'' \emph{ACM
  conference on Computer and communications security}, pp. 116--127, 2007.

\bibitem{sobel1971noncooperative}
M.~J. Sobel, ``Noncooperative stochastic games,'' \emph{The Annals of
  Mathematical Statistics}, vol.~42, no.~6, pp. 1930--1935, 1971.

\bibitem{Orl:93}
J.~B. Orlin, ``A faster strongly polynomial minimum cost flow algorithm,''
  \emph{Operations Research}, vol.~41, no.~2, pp. 338--350, 1993.

\bibitem{holmes}
S.~M. Milajerdi, R.~Gjomemo, B.~Eshete, R.~Sekar, and V.~Venkatakrishnan,
  ``Holmes: {R}eal-time {APT} detection through correlation of suspicious
  information flows,'' \emph{IEEE Symposium on Security and Privacy}, pp.
  1137--1152, 2019.

\bibitem{heuristic-conf}
T.~Chen, L.-A. Tang, Y.~Sun, Z.~Chen, and K.~Zhang, ``Entity embedding-based
  anomaly detection for heterogeneous categorical events,'' \emph{International
  Joint Conference on Artificial Intelligence}, pp. 1396--1403, 2016.

\end{thebibliography}

\end{document}